\documentclass[11pt,a4paper]{article}
\usepackage{empheq}
\usepackage{geometry}
\usepackage{tikz}
\usepackage{amsthm}
\usepackage{amsmath}
\usepackage{amsmath,bm}
\usepackage{graphicx}
\usepackage{bbm}
\usepackage{booktabs}
\usepackage{graphicx} 
\usepackage{amsfonts}
\usetikzlibrary{scopes,matrix,positioning}
\usepackage{algorithm,algorithmic}
\usepackage{blindtext}
\usepackage{amssymb}
\usepackage{authblk}
\usepackage{mathtools}
\usepackage{caption}
\usepackage{textcomp}
\usepackage{enumitem}
\usepackage{hyperref}
\usepackage{footnote}
\usepackage{threeparttable}
\usepackage[titletoc]{appendix}
\usepackage{authblk}
\usepackage{setspace} 
\usepackage{subcaption}
\DeclarePairedDelimiter{\ceil}{\lceil}{\rceil}
\DeclarePairedDelimiter\floor{\lfloor}{\rfloor}

\providecommand{\keywords}[1]{\textbf{\textit{Keywords:}} #1}

\newtheorem{theorem}{Theorem}

\newtheorem{lemma}{Lemma}
\newtheorem{remark}{Remark}
\newtheorem{proposition}{Proposition}

\newtheorem{definition}{Definition}
\newtheorem{example}{Example}

\makeatletter
\renewcommand\@biblabel[1]{}
\renewenvironment{thebibliography}[1]
     {\section*{\refname}%
      \@mkboth{\MakeUppercase\refname}{\MakeUppercase\refname}%
      \list{}%
           {\leftmargin0pt
            \@openbib@code
            \usecounter{enumiv}}%
      \sloppy
      \clubpenalty4000
      \@clubpenalty \clubpenalty
      \widowpenalty4000%
      \sfcode`\.\@m}
     {\def\@noitemerr
       {\@latex@warning{Empty `thebibliography' environment}}%
      \endlist}
\makeatother

\tikzset{
  mymx/.style={matrix of math nodes,nodes=myball,column sep=4em,row sep=-1ex},
  myball/.style={draw,circle,inner sep=0pt},
  mylabel/.style={midway,sloped,fill=white,inner sep=1pt,outer sep=1pt,below,
    execute at begin node={$\scriptstyle},execute at end node={$}},
  plain/.style={draw=none,fill=none},
  sel/.append style={fill=green!10},
  prevsel/.append style={fill=red!10},
  route/.style={-latex,thick},
  selroute/.style={route,blue!50!green}
}

\tikzset{   cirwhite/.style={draw=gray,circle,fill=white,minimum size=1pt,inner sep=2pt,line width=0.2mm},
            cirred/.style={draw=gray,circle,fill=red,minimum size=1pt,inner sep=2pt,line width=0.2mm},
}

\author[1]{Babak Mahdavi-Damghani}
\author[2]{Konul Mustafayeva}
\author[2]{Cristin Buescu}
\author[1]{Stephen Roberts}

\affil[1]{Oxford-Man Institute of Quantitative Finance, Oxford, UK}

\affil[2]{Department of Mathematics, King's College London, London, UK \vspace{3ex}}

\date{\vspace{-3ex}}
 \title{Portfolio Optimization for Cointelated Pairs: SDEs vs Machine Learning\vspace{2ex}}

\begin{document}
\maketitle  

\begin{abstract}
With the recent rise of Machine Learning as a candidate to partially replace classic Financial Mathematics methodologies, we investigate the performances of both in solving the problem of dynamic portfolio optimization in continuous-time, finite-horizon setting for a portfolio of two assets that are intertwined. 

In Financial Mathematics approach we model the asset prices not via the common approaches used in pairs trading such as a high correlation or cointegration, but with the cointelation model in \cite{CointAndInferred} that aims to reconcile both short-term risk and long-term equilibrium. We maximize the overall P\&L with Financial Mathematics approach that dynamically switches between a mean-variance optimal strategy and a power utility maximizing strategy. We use a stochastic control formulation of the problem of power utility maximization and solve numerically the resulting HJB equation with the Deep Galerkin method introduced in \cite{DGM}. 

We turn to Machine Learning for the same P\&L maximization problem and use clustering analysis to devise bands, combined with in-band optimization. Although this approach is model agnostic, results obtained with data simulated from the same cointelation model as FM give an edge to ML. 
\end{abstract}
\keywords{Pairs Trading, Cointelation, Portfolio Optimization, Stochastic Control, \\
Band-wise Gaussian Mixture, Deep Learning.\vspace{22ex}}

\noindent Correspondence to: Konul Mustafayeva, Department of Mathematics, King's College London, Strand, London,  WC2R 2LS. \\
E-mail: konul.mustafayeva@kcl.ac.uk 
81

\thispagestyle{empty}
\clearpage

\section{Introduction} 
In a financial market with two assets that exhibit clear dependence we go beyond high correlation which is used in pairs trading, and model asset prices with the hybrid model of cointelation introduced in \cite{CointAndInferred}. We solve the portfolio optimization problem employing a more general set of admissible strategies than long/short strategies used in pairs trading. 

A pairs trading strategy involves matching a long position with a short position in two assets with a high correlation. Pairs trading was pioneered in the mid 1980s by a group of quantitative researchers from Morgan Stanley (for introduction to pairs trading see \cite{Vidyamurthy}). The securities in a pairs trade must have a high positive correlation, which is the primary driver behind the strategy’s profits. 

Pairs trading is based on the high historical correlation of two assets and a trader's view that the two securities will maintain a specified correlation. A pairs trading strategy is applied when a trader identifies a correlation discrepancy. More specifically, the trader monitors performance of two historically correlated securities. When the correlation between the two securities temporarily weakens, i.e. the spread widens, the trader applies a trading strategy which shorts the high asset and buys the low asset. As the spread narrows again to some equilibrium value, a profit results. 

However, many authors argue that correlation is an inappropriate measure of dependency in financial markets, since returns often exhibit a nonlinear co-dependence (e.g. \cite{Alexander}, \cite{Wilmott}). Mahdavi-Damghani, et al. \cite{Damgani} showed that measured correlation of the returns of a mean-reverting processes is misleading: a strong positive correlation does not necessarily imply that two stochastic processes move in the same direction and vice versa. Cointegration, on the other hand, tests the long-term equilibrium relationships between assets and has been extensively used in pairs trading ( see \cite{Vidyamurthy}). Cointegration tests do not measure how well two variables move together, but rather whether the difference between their means remains constant. Sometimes series with high correlation will also be cointegrated, and vice versa, but this is not always the case.

The cointelation model was introduced in \cite{CointAndInferred} as a hybrid model which reconciles correlation and cointegration by capturing both short-term risk and long-term equilibrium. The rationale for the long term risk is that during the time of rare market crashes all assets prices fall. However, in the more bullish periods, the short term risk increases, the long term risk becomes less pronounced and the “macro” driver less visible. These influences are accompanied with mean reversion forces from one asset to the other.  

In this setting we consider a continuous-time, finite horizon portfolio optimization problems for pairs of assets whose prices follow the cointelation model in \cite{CointAndInferred}. Generally, the optimization problem is to find the optimal control
\begin{equation}
   \tilde{w}^* = \underset{\tilde{w} \in \mathcal{A}}{\text{argmax}} \ \ U(X^{\tilde{w}}_t, Y^{\tilde{w}}_t) \label{eq:generalOptimProb}
\end{equation}
where $U(x)$ is a utility function, $\tilde{w} = (\tilde{w}_1, \tilde{w}_2)$ is a vector of proportions of wealth invested in each asset, $\mathcal{A}$ is a set of admissible strategies: either $\tilde{w}_1 = - \tilde{w}_2$ (long/short) or $\tilde{w}_1, \tilde{w}_2 > 0$ with $\tilde{w}_1 + \tilde{w}_2 = 1$ (long only).

We solve the portfolio optimization problem in \eqref{eq:generalOptimProb} with Financial Mathematics and Machine Learning methodologies and compare their performance. In Financial Mathematics approach we use SDE evolution of asset prices, whereas the Machine Learning approach does not assume an underlying model and applies generally to any pair of assets. 

In Section \ref{sec:CointelationModel} we review the cointelation model. In Section \ref{sec:FM_approach} we use the classical Financial Mathematics criteria: mean-variance optimization and power utility maximization. In Section \ref{sec:ML_approach} we use clustering analysis from Machine Learning to solve the P\&L maximization problem. We present the results of each approach in Section \ref{sec:resultsAndConclusion} and discuss them comparatively.
 
\section{Review of cointelation model for pairs of asset }\label{sec:CointelationModel}
We first present the usual way correlation is calculated in the financial industry (see e.g. p.274 \cite{Wilmott}, \cite{Alexander}). Assume we have two assets with prices modeled by stochastic processes $(X_t)_{t\geq 0}$ and $(Y_t)_{t\geq 0}$ on a probability space $(\Omega , \mathcal{F}, \mathbb{P})$. We have $N$ observations of $X$ and $Y$ at intervals $\Delta t$, i.e. $X(t_i)$ and $Y(t_i)$ with $i=1,...,N$ and $\Delta t = t_i - t_{i-1}$. Here  $\Delta t \in \{1, 5, 22, 252\}$ corresponds to daily, weekly, monthly and yearly data. The $\Delta t$-returns on $i$-th data point of assets $X$ and $Y$ is
\begin{eqnarray}
    R_X(t_i, \Delta t) &=& \frac{X(t_i + \Delta t)  - X(t_i)}{X(t_i)} \\
    R_Y(t_i, \Delta t) &=& \frac{Y(t_i + \Delta t) - Y(t_i)}{Y(t_i)}.
\end{eqnarray}
The sample volatilities of time series of asset prices $X$ and $Y$ are then
\begin{eqnarray}
 \sigma_X(\Delta t) &=& \sqrt{\frac{1}{\Delta t(N-1)}\sum_{i=1}^N (R_X(t_i, \Delta t) - \Bar{R}_X)^2} \\
 \sigma_Y(\Delta t) &=& \sqrt{\frac{1}{\Delta t(N-1)}\sum_{i=1}^N (R_Y(t_i, \Delta t) - \Bar{R}_Y)^2}, 
\end{eqnarray}
where $\Bar{R}_X, \Bar{R}_Y$ are the sample average of all the returns in the series of $X$  and $Y$, respectively. 

The sample covariance between the returns of assets $X$ and $Y$ is given by
\begin{equation}
 \sigma_{XY}(\Delta t) = \frac{1}{\Delta t(N-1)}\sum_{i=1}^N (R_X(t_i, \Delta t) - \Bar{R}_X)(R_Y(t_i, \Delta t) - \Bar{R}_Y).
\end{equation}
In this paper we consider the measured correlation, which is the sample cross-correlation given by
\begin{equation}
\rho_{XY}(\Delta t) = \frac{\sigma_{XY}(\Delta t)}{\sigma_X(\Delta t) \sigma_Y(\Delta t)}. \label{eq:measuredCorr}
\end{equation}

For correlation to be an appropriate choice of measure of co-dependence the assumption of linear dependency between series needs to be satisfied (see Chapter 1.4 \cite{Alexander}). Often in financial markets with a non-linear dependence between returns, the correlation is an inappropriate measure of co-dependency and is misleading, especially when used to capture long-term relationship between assets (see \cite{Damgani} and \cite{CointAndInferred}).  

An alternative statistical measure to correlation is cointegration. If two time series $X_t$ and $Y_t$ are integrated\footnote{A time series $X_t$ is integrated of order $d$ if $(1-L)^dX_t$ is a stationary process. Here $L$ is a lag operator.} of order $d$ and there exists $\beta$ such that a linear combination $X_t + \beta Y_t$ is integrated of order less that $d$, then $X_t$ and $Y_t$ are cointegrated (see \cite{EngleGranger}). Since the spread of cointegrated asset prices is mean reverting, they have a common stochastic trend, i.e. the asset prices are ‘tied together’ in the long term, although they might drift apart in the short-term (see \cite{Alexander_Optim_heging}). Because the cointegration requires sophisticated statistical analysis, it has not been used as widely as correlation in the financial industry. 


Although correlation and cointegration are related, they are different concepts. High correlation does not necessarily imply high cointegration, and neither does high cointegration imply high correlation (e.g. see Figure 4 in \cite{Damgani}). Two assets may be perfectly correlated over short timescales yet diverge in the long run, with one growing and the other decaying. Conversely, two assets may follow each other, with a certain finite spread, but with any correlation, positive, negative or varying. 

Mahdavi-Damghani \cite{CointAndInferred} proposed cointelation as a hybrid model that aims to mediate between correlation and cointegration. It captures both short-term and long-terms relationships between the assets. 

\begin{definition} \label{def:cointmodel} Consider a filtered probability space by $(\Omega, \mathcal{F}, {(\mathcal{F}_t)}_{(t \geq 0)}, \mathbb{P})$, with the historical probability measure, $\mathbb{P}$. The cointelation model for a pairs of assets with prices $X_t$ and $Y_t$ defined in \cite{CointAndInferred} as 
\begin{eqnarray}\nonumber
&& dX_{t} = \mu X_t dt + \sigma X_t dW_t,   \\\nonumber
&& dY_{t} = \kappa(X_{t}-Y_{t})dt  + \eta Y_{t} d\tilde{W}_t, \\ 
&& d \langle W, \tilde{W} \rangle_t = \rho dt, \label{eq:modelCointel}
\end{eqnarray}
where $\mu \in \mathbb{R}$, $\sigma>0$, $X(t_0) = x_0$ are the drift, diffusion coefficients and initial value of asset price $X$; $0< \kappa \leq 1$, $\eta > 0$, $Y(t_0) = y_0 > 0$ are the rate of mean reversion, volatility and initial value of the asset price $Y$; $(\tilde{W}(t))_{t\geq 0}$ and $(W(t))_{t\geq 0}$ are two correlated Brownian motions with constant correlation coefficient $-1 \leq \rho \leq 1$ that generate the filtration ${(\mathcal{F}_t)}_{(t \geq 0)}$.
\end{definition}
The processes $(X)_{t\geq 0}$ and $(Y)_{t\geq 0}$ are called \textit{the leading process} and \textit{the lagging process}, respectively. This is due to the fact that the lagging process reverts around the leading process.

We present here the concepts of \textit{inferred correlation} function and \textit{number of crosses} formula introduced in \cite{CointAndInferred} in order to device a test whether two pairs are cointelated .  

Let $\rho_{XY}^*(\Delta t)$ be the inferred correlation function between two times series of cointelated asset prices defined as follows
\begin{equation}
\rho^{*}_{XY}(\Delta t) = \displaystyle \sup_{0<\tilde{\Delta}t\leq \Delta t}\rho_{XY}(\tilde{\Delta} t).  \label{eq:inferredCorrDefinition}
\end{equation}
Sometimes there may not be enough data to calculate $\Delta t$-inferred (measured) correlation of cointelated assets. In \cite{CointAndInferred} the following formula for approximation of inferred correlation \eqref{eq:inferredCorrDefinition} was proposed via examining various data sets:
\begin{equation}
\rho^{*}_{XY}(\Delta t) \approx  \rho + (1-\rho) \left[1 - \exp\left(- \lambda\kappa(\Delta t -1)\right)\right],  \label{eq:InferredCorr}
\end{equation}
where  $\kappa \in [0,1]$, $\lambda >0$, $\rho \in [-1,1]$. The parameter $\lambda \approx 1.75$ for "regular financial data", although it is itself a function in general. Thus, if one does not have enough empirical data to calculate, for example, the yearly (252 days) inferred correlation, the formula in equation \eqref{eq:InferredCorr} allows to approximate it using only $\kappa$ and $\rho$ parameters of cointelation model in \eqref{eq:modelCointel} and setting $\Delta t =252$, $\lambda \approx 1.75$.

The motivation for inferred correlation approximation formula \eqref{eq:InferredCorr}, is that in the discrete version of the processes in equation \eqref{eq:modelCointel} the measured correlation increases as the time increment, $\Delta t$, increases (e.g. correlations calculated using daily, weekly, monthly returns). Moreover, the measured correlation of cointelated pairs will converge to $1$ faster as the speed of mean reversion parameter $\kappa$ increases. If we set $\rho=-1$ in \eqref{eq:modelCointel}, the inferred correlation of cointelated asset prices may cover the whole correlation spectrum $[-1,1]$ (see Figure \ref{fig:CointelationWithRhoM1}).

Another way for testing if two times series are cointelated is to study how many times the normalized series cross paths. If one discretizes equation \eqref{eq:modelCointel}, then one can approximate the expectation of the number of times, $\Gamma_{x,y}$, the two stochastic process, $x=X_{i\in[1,2,\ldots N]}$ and $y=Y_{i\in[1,2,\ldots N]}$, cross paths as follows
\begin{equation}
E[\Gamma_{x,y}(\kappa,N)] \approx  N\left[\gamma(1-\kappa) + \frac{1}{2} \sqrt{\kappa} \right] \label{eq:rollingCrosses}
\end{equation}
with $N$ is the length of the data, $\gamma$ is a positive constant and $\kappa$ is the speed of mean reversion in equation \eqref{eq:modelCointel}.

Compared to the number of times purely correlated SDEs (eg: without the mean reversion component, i.e. when $\kappa = 0$) the number of times the discrete version of the cointelated SDEs cross paths is larger than if they were random, and the bigger the $\kappa$ the more often the paths of discretized SDEs \textit{cross each other} per unit of time. 

Then two stochastic processes are cointelated (see \cite{CointAndInferred}) if
\begin{itemize}
    \item Inferred correlation formula in equation \eqref{eq:InferredCorr} is verified;
    \item The number of crosses formula in equation \eqref{eq:rollingCrosses} is verified;
    \item the underlying assets have a reasonable physical connection that would suggest their spread should mean revert (e.g. oil and BP share prices). 
\end{itemize}

The parameters in cointelation model \eqref{eq:modelCointel} can be estimated using the inferred correlation formula \eqref{eq:InferredCorr} and the number of crosses formula \eqref{eq:rollingCrosses} (see \cite{CointAndInferred}). Similarly to the variance reduction methodology described in \cite{Damgani}, \cite{CointAndInferred}, we can define 
\begin{eqnarray}
B_{+} &=& \left|\frac{\max(X_{t}-Y_{t},t\in[0,T])}{2}\right|, \\ \label{eq:Bplus}
B_{-} &=& \left|\frac{\inf(X_{t}- Y_{t},t\in[0,T])}{2}\right|.
\label{eq:Bmminus}\nonumber
\end{eqnarray}
We note that the estimation of $\kappa$ has a higher variance when \begin{equation}
Z_{\rho} = B_{+}>|X_{t}-Y_{t}|>B_{-}, \label{eq:Zsigma}
\end{equation}
where $\rho$, on the other hand has quality samples. The reverse is true when
\begin{equation}
Z_{\kappa}=|X_{t}-Y_{t}|>B_{+} \bigcup |X_{t}-Y_{t}|<B_{-}. \label{eq:Ztheta}
\end{equation}
We can therefore sample  $\kappa$ in $Z_{\kappa}$ and $\rho$ in $Z_{\rho}$. Figure \ref{fig:BplusBminus} illustrates this.




\section{Financial Mathematics approach for portfolio optimization problem} \label{sec:FM_approach}
We consider the portfolio of two assets and model the their prices with the cointelation in \eqref{eq:modelCointel}. We approach the optimization problem of this portfolio with classic Financial Mathematics criteria: mean-variance and power utility maximization. Since the cointelated assets are characterized by both correlation and mean-reversion components, we formulate the mean-variance optimization problem for long only strategies and we calculate the optimal strategies to make profit on correlation. To make profit on mean-reversion property of the cointelated assets we use stochastic control formulation of the power utility maximization problem for long/short strategies and calculate the optimal weights. We then maximize portfolio P\&L by dynamically switching between these two optimal strategies.

\subsection{Mean-variance optimization}\label{sec:LitReviewMarkowitz}
We first review fundamental notions and concepts for mean-variance optimization.
\paragraph{Returns:}
A portfolio considers a combination of $n$ potential assets, with an initial capital $V(0)$ and weights $w_1, w_2,..., w_n$, such that $\sum_i^n w_i =1$, $w_iV(0)$ is the amount invested in security $i$ for $i= 1,2,...,n$ at time $t=0$.  
The number of shares to invest in security $i$ at time $t=0$ is 
\begin{equation}
n_i = \frac{w_iV(0)}{S_i(0)}.
\end{equation}
The value of portfolio at time $t$ is
\begin{equation}
V(t) = \sum_{i=1}^N n_i S_i(t).
\end{equation}
Given the number of shares $n_i$ with $i = 1,...,n$, the percentage of the portfolio invested in asset $i$ at time $t$ is
\begin{equation}
w_i(t) = \frac{n_iS_i(t)}{\sum_{i=1}^N n_i S_i(t)}, \label{eq:percentTotalwealth}
\end{equation}
with $\sum_{i=1}^N w_i(t) = 1$.
The rate of return of asset $i$ at time $t$ (i.e. over $[t -\Delta t,t]$) is given by
\begin{equation}
R_i(t) = \frac{S_i(t) - S_i(t-\Delta t)}{S_i(t-\Delta t)} = \frac{S_i(t)}{S_i(t-\Delta t)} - 1.
\end{equation}
The rate of return of portfolio, $R_p(t)$, is then
\begin{equation}
R_p(t) = \frac{V(t) - V(t-\Delta t)}{V(t-\Delta t)}.
\end{equation}
We can show that the return of portfolio is a linear combination of the returns of individual assets as follows
\begin{eqnarray} \nonumber
R_p(t) &=& \frac{V(t)}{V(t-\Delta t)} - 1 = \sum_{i=1}^N \frac{n_i S_i(t)}{\sum_{j=1}^N n_i S_i(t-\Delta t)} - 1 \\ \nonumber
&=& \sum_{i=1}^N \frac{n_i S_i(t-\Delta t)S_i(t)}{\sum_{j=1}^N n_i S_i(t-\Delta t)S_i(t-\Delta t)} - 1 = \sum_{i=1}^N w_i(t)(R_i(t)+1) - 1 \\ 
&=&  \sum_{i=1}^N w_i(t)R_i(t). \label{eq:rateOfreturn}
\end{eqnarray}
Sometimes it is more convenient to use log returns, which are defined for asset $i$ by 
\begin{equation}
r_i(t) = \ln\left(\frac{S_i(t)}{S_i(t-\Delta t)}\right).
\end{equation}
It should be pointed out that for short period of time the log return is approximately equal to the rate of return 
\begin{equation}
r_i(t) = \ln\left(\frac{S_i(t)}{S_i(t-\Delta t)}\right) = \ln(R_i(t) +1) \approx R_i(t). \label{eq:logReturn}
\end{equation}
Therefore we do not distinguish between these two returns, as long as the time increment, $\Delta t$, is short. Going forward we will use daily logarithmic returns. Thus, the return of portfolio, $r_p$, at time at time $t$ in this case becomes
\begin{equation}
r_p = \sum_{i=1}^N w_ir_i.  
\end{equation}
\paragraph{Expectation and variance of returns:}
By the linearity property of expected value operator, the expected return of portfolio, $E(r_p)$, is 
\begin{eqnarray} 
E(r_p) = E\left( \sum_{i=1}^N w_i r_i \right) = \sum_{i=1}^N w_i E(r_i) = \sum_{i=1}^N w_i \mu_i = w^\top \mu , 
\end{eqnarray}
where $\mu_i$ denotes the expected return of asset $i$ and $w^\top = [w_1, w_2, ..., w_n]$, $\mu = [\mu_1, \mu_2, ..., \mu_n]^\top$.

The variance of the return of portfolio, $Var(r_p)$, is given by
\begin{eqnarray} \nonumber
Var(r_p) &=& E \left[ \left(\sum_{i=1}^N w_i r_i - E(r_p)\right)^2  \right] = E \left[\left( \sum_{i=1}^N w_i(r_i - E(r_i)) \right)^2 \right]\\ \nonumber
&=& E \left[ \left( \sum_{i=1}^N w_i(r_i - E(r_i)) \right) \left(  \sum_{j=1}^N w_j(r_j - E(r_j) )\right) \right] \\ \nonumber
 &=& \sum_{i=1}^N \sum_{j=1}^N w_iw_j \underbrace{E\left[ (r_i - E(r_i))(r_j - E(r_j))\right]}_{:= \sigma(r_i,r_j)}\\
&=& \sum_{i=1}^N \sum_{j=1}^N w_iw_j \sigma(r_i,r_j) = w^\top \Sigma w,
\end{eqnarray} 
where $\Sigma$ denotes the covariance matrix of the asset returns, composed of all covariances between the returns of assets $i$ and $j$ defined as  $\sigma(r_i,r_j)$. The variance of asset $i$’s return, which 
constitute the diagonal of the covariance matrix, is $\sigma(r_i,r_i)$.

\subsubsection*{Optimal investment strategy using mean-variance criterion}\label{sec:MVS}
We consider a portfolio consisting of two assets. The uncertainty is modelled by a probability space $(\Omega, \mathcal{F}, P)$ with a filtration $(\mathcal{F}_t)_{t\geq 0}$ generated by two-dimensional Brownian
motion: $(W, \tilde{W})$. Denote by $X(t)$ and $Y(t)$ the prices of two assets at time $t$, with dynamics following cointelation model in \eqref{eq:modelCointel}.
The investment behavior is modelled by an investment strategy $h = (h_1, h_2)$. Here, $h_i \in [0,1]$, $i = 1, 2$, denotes the percentage of total wealth invested in $i$-th asset (see equation \eqref{eq:percentTotalwealth}). Let $h_1(t)$ and $h_2(t)$ denote respectively the portfolio weights for assets $X$ and $Y$ at time $t$. The holdings are allowed to be adjusted continuously up to a fixed horizon $T$.  

Denoting by $V^h_t$ the value of portfolio at time $t$ associated to a strategy $h$ we have
\begin{equation}
V^h(t) = \frac{h_1(t)V^h(t)}{X(t)}X(t) + \frac{h_2(t)V^h(t)}{Y(t)}Y(t) , \\ \label{eq:portvalh}
\end{equation}
with initial wealth $V^h(t_0) = v_0$. We restrict our considerations to self-financing strategies, where the value of the portfolio changes only because the asset prices change, i.e. there is no inflow or withdrawal of money \cite{HarrisonKreps}. In this case the dynamic of the wealth process is 
\begin{equation}
dV^h(t)=V^h(t)\left[ h_1(t)\frac{dX(t)}{X(t)} + h_2(t)\frac{dY(t)}{Y(t)} \right]. 
\end{equation}

Let $\mathcal{A}^1$ denote the set of all \textit{admissible strategies}, $h = (h_1, h_2)$, satisfying:
\begin{itemize}
\item[(i)] Given $v_0 > 0$ the wealth process $V^{v_{0},h}(\cdot)$ corresponding to $w_0,h$ satisfies 
\begin{equation}
V^{v_0,h}(t) \geq 0, \ \ \ 0 \leq t \leq T,
\end{equation}
\item[(ii)] $h_i(t) \geq 0$ for all $i=1,2$,
\item[(iii)] $\sum_{i=1}^2h_i(t) = 1$. 
\end{itemize}

An investment strategy, $h \in \mathcal{A}^1$, is called optimal if there exists no other strategy $\tilde{h} \in \mathcal{A}^1$ such that $E(r_p(h)) \geq E(r_p(\Tilde{h)})$ and $Var(r(h)) \leq Var(r(\Tilde{h}))$ with at least one inequality being strict (see \cite{LiNg}).

We define a utility function, $U(t,h)$, as in \cite{Bodie}:
\begin{equation}
    U(t,h) = 2\tau E[r_p(t)] - \sigma^2[r_p(t)], \label{eq:utilityMVC}
\end{equation}
where $\tau \geq 0$ is the risk tolerance coefficient. Then according to \cite{Garciaa} we have the following proposition.
\begin{proposition}[Mean-Variance Criterion]
Finding an optimal strategy for mean-variance criteria is equivalent to the utility maximization problem:
\begin{equation}
\max_{h(t)} U(t, h) \label{eq:MVCproblem}
\end{equation}
with constraints
\begin{itemize}
\item $\sum_{i=1}^N h_i = 1$,
\item $h_i \geq 0 \ \ \forall i$.
\end{itemize}
and $U(t, h)$ given in \eqref{eq:utilityMVC}.
\end{proposition}
Thus we have optimization problem in equation \eqref{eq:MVCproblem}. 
From equation \eqref{eq:rateOfreturn} we have that the rate of return of our portfolio, $R_p$, over $[t - \Delta t,t]$ is  
\begin{equation}
R_p(t) = \frac{V^h(t) - V^h(t-\Delta t)}{V^h(t - \Delta t)} = \sum_{i=1}^2h_i(t)R_i(t),
\end{equation}
where $R_i$ is the rate of return of individual assets. The log return of our portfolio, $r_p$ is given by
\begin{equation}
r_p(t) = h_1r_1(t) + h_2r_2(t), \label{eq:portReturn}
\end{equation}
where $r_i(t) \approx R_i(t)$, as we showed in equation \eqref{eq:logReturn}.
\begin{lemma} \label{lem:valueOfPortfolio}
Denote by $V^{h}(t)$ the value of the portfolio corresponding to the admissible strategy $h\in \mathcal{A}^1$. Then:
\begin{itemize}
\item[(i)] The expectation of portfolio return over $[t- \Delta t, t]$ is 
\begin{equation}
E(r_p(t)) = h_1E[r_X(t)] + h_2E[r_Y(t)]. 
\vspace{-\baselineskip}
\end{equation}\label{eq:portfolio_expectation}
\item[(ii)] The variance of portfolio return over $[t - \Delta ,t]$ is
\begin{eqnarray} 
Var(r_p(t)) &=& h_1^2Var[r_X(t)] + h_2^2Var[r_Y(t)] + 2h_1h_2Cov[r_Xr_Y(t)],
\vspace{-\baselineskip}
\end{eqnarray}\label{eq:portfolio_variance}
\vspace{-\baselineskip}
\end{itemize}
where $r(X_t)=\ln\left( \frac{X_t}{X_{t-\Delta t}}\right)$ and $r(Y_t)=\ln\left( \frac{Y_t}{Y_{t-\Delta t}}\right)$ the daily log returns of assets $X$ and $Y$ and
\begin{itemize}
    \item $E(r_X(t)) = (\mu - \frac{\sigma^2}{2})\Delta t$ is the expected return of the asset price $X$  over the horizon $[t - \Delta t,t]$;
    \item $E(r_Y(t)) = [ \ln\left( ae^{\mu \Delta t} + (Y_0 -a)e^{-\kappa \Delta t}\right) - \frac{ ce^{(2\mu+\sigma^2)\Delta t}  + de^{(\mu -\kappa +\sigma\eta\rho)\Delta t}}{2(ae^{\mu \Delta t} + (Y_0 - a)e^{-\kappa \Delta t})^2}
- \ln(Y_{t-\Delta t})] - $\\
$ \frac{(Y_0^2-c-d)e^{2(\eta^2 - \kappa)\Delta t}}{2(ae^{\mu t} + (Y_0 - a)e^{-\kappa \Delta t})^2} + \frac{1}{2}$ is the expected return of the asset price $Y$  over the horizon $[t - \Delta t,t]$;
    \item $Var(r_X(t)) = \sigma^2\Delta t$ is the variance of return of asset price $X$ over the horizon $[t -\Delta t, t]$;
    \item $Var(r_Y(t)) = \frac{ ce^{(2\mu+\sigma^2)\Delta t}}{(ae^{\mu \Delta t} + (Y_0 - a)e^{-\kappa \Delta t})^2} + \frac{de^{(\mu+\sigma\eta\rho -\kappa)\Delta t}}{(ae^{\mu \Delta t} + (Y_0 - a)e^{-\kappa \Delta t})^2}
 +\frac{(Y_0^2-c-d)e^{2(\eta^2 - \kappa)\Delta t}}{(ae^{\mu \Delta t} + (Y_0 - a)e^{-\kappa \Delta t})^2} - 1$ is the variance of return of asset price $Y$ over the horizon $[t -\Delta t, t]$;
    \item $Cov(r_X(t)r_Y(t)) = \ln\left(\frac{be^{(\mu+\sigma^2)\Delta t} + (X_0Y_0 - b)e^{(\sigma\eta\rho - \kappa)\Delta t}}{aX_0e^{2\mu \Delta t} + (X_0Y_0- aX_0)e^{(\mu -\kappa)\Delta t}}\right)$ is the covariance of returns of two asset prices $X$ and $Y$ over the horizon $[t -\Delta t, t]$.
\end{itemize} 
\end{lemma}
\begin{proof}
See Appendix \ref{appendix:MVC}
\end{proof}
The optimal weights for mean-variance criterion were derived in \cite{Soeryana}. We state the following proposition from \cite{Soeryana} applied to the cointelation model \eqref{eq:modelCointel}.
\begin{proposition}
The optimal solution for the problem in \eqref{eq:MVCproblem} for cointelation model \eqref{eq:modelCointel} is:
\begin{eqnarray}
h^*(t) = \frac{1}{e'\Sigma^{-1}(t)e}\Sigma^{-1}(t)e + \tau \left[ \Sigma^{-1}(t)M(t) - \frac{e'\Sigma^{-1}(t)M(t)}{e'\Sigma^{-1}(t)e}\Sigma^{-1}(t)e \right], \label{eq:optimalweighMVC}
\end{eqnarray}
with $e' = [1,1,1]$, $M(t)=[E(r_X(t)), E(r_Y(t))]$ and covariance matrix is
\begin{equation*} 
\Sigma(t) = 
  \begin{bmatrix}
    Var(r_X(t))  & Cov(r_X(t),r_Y(t))  \\
    Cov(r_X(t),r_Y(t)) & Var(r_Y(t)) 
\end{bmatrix},  
\end{equation*}
where the expressions for $E(r_X(t))$, $E(r_Y(t))$ and  $Var(r_X(t))$, $Var(r_Y(t))$, $Cov[r_X(t),r_Y(t)]$ are given above.
\end{proposition}
Replacing these formulas for expectation, variance and covariance of the returns of asset prices in equation \eqref{eq:optimalweighMVC}, we get optimal strategies for mean-variance optimization problem. We will present numerical examples in Section \ref{sec:resultsAndConclusion}.

\subsection{Stochastic control for pairs trading \label{sec:stochasticControl}} 
\subsubsection*{Power utility maximization problem}
We now use a stochastic control approach to the power utility maximization problem. Here we mainly follow \cite{MudchanatongsukPrimbsWong}, but with modified dynamics for asset prices. More specifically, they assume the price dynamics of one of the assets is a geometric Brownian motion and model the log-spread as an Ornstein-Uhlenbeck process. We, however, assume the dynamics of asset prices are governed by the cointelation model in equation \eqref{eq:modelCointel}, where one of the assets follow the geometric Brownian motion and the second asset mean reverts around the first one. 

Let $(\Omega, \mathcal{F}, P)$ be a complete probability space with a filtration $(\mathcal{F}_t)_{t\geq 0}$ generated by two-dimensional Brownian motion: $(W, \tilde{W})$. We consider the same market as in Subsection \ref{sec:MVS}: two assets which follow the cointelation model \eqref{eq:modelCointel}.

We assume an initial wealth $v_0 > 0$ at time $t=0$. Initial wealth is held in a margin account. For simplicity we assume that the interest rate for margin account is 0, $r=0$.  Margin account restricts how much one can short or long. The holdings are allowed to be adjusted continuously up to a fixed horizon $T$. The investment behavior is modelled by an investment strategy $\pi = (\pi_1, \pi_2)$. Here, $\pi_i(t)$, $i = 1, 2$, denotes the percentage of total wealth invested in $i$-th asset at time $t$ (see equation \eqref{eq:percentTotalwealth}). Let $\pi_1(t)$, $\pi_2(t)$ be respectively the portfolio weights for assets $X$ and $Y$ at time $t$. We only allow pairs trading: short one of the asset and long the other in equal dollar amount, i.e. $\pi_1(t) = -\pi_2(t)$. In addition, we restrict our considerations to self-financing strategies.

We define admissible control and controlled process as in \cite{KornKraft}.
\begin{definition}[Control]
Given a subset $U$ of $\mathbb{R}^2$, we denote by $\mathcal{U}_0$ the set of all progressively measurable processes $\pi = \{ \pi_t,t\geq 0\}$ valued in $U$. The elements of $\mathcal{U}_0$ are called control processes.
\end{definition}
Denote by $V^\pi(t)$ the value of portfolio corresponding to strategy $\pi$ at time $t$, which is given by 
\begin{equation}
V^\pi(t) = \frac{\pi_1(t)V^\pi(t)}{X(t)}X(t) + \frac{\pi_2(t)V^\pi(t)}{Y(t)}Y(t). \label{eq:portvalpi}
\end{equation}
The dynamics of the portfolio value $V^\pi$ associated with strategy $\pi = (\pi_1, \pi_2)$ is given by
\begin{equation}
dV^\pi(t)=V^\pi(t)\left[ \pi_1(t)\frac{dX(t)}{X(t)} + \pi_2(t)\frac{dY(t)}{Y(t)} \right]
\label{eq:WD}
\end{equation}
Replacing the dynamics for $X(t)$ and $Y(t)$ into \eqref{eq:WD} we get:
\begin{equation}
dV^\pi(t) = V^\pi(t)\Big[ \pi_1(\mu dt +\sigma dW(t)) - \pi_1\left(\kappa\left(\frac{X(t)}{Y(t)} -1\right)dt + \eta d\tilde{W}(t)\right)\Big]. \\
\end{equation}
\begin{lemma}\label{prop:logspread_dynamics}
Denote $Z(t):=\frac{X(t)}{Y(t)}$. For the cointelation model \eqref{eq:modelCointel} we obtain that $Z(t)$ has the dynamics
\begin{eqnarray}
dZ(t) = [\mu +\eta^2 -\sigma\eta\rho -\kappa(Z(t) -1)]Z(t)dt + Z(t)(\sigma dW(t) +\eta d\tilde{W}(t)). \label{eq:ZDynamics}
\end{eqnarray}
\end{lemma}
\begin{proof}
By Ito's quotient rule:
\begin{equation}
d\left(\frac{X(t)}{Y(t)}\right) = \frac{dX(t)}{X(t)}\frac{X(t)}{Y(t)} - \frac{dY(t)}{Y(t)}\frac{X(t)}{Y(t)} +  \frac{d\langle Y,Y\rangle_t}{Y(t)^2}\frac{X(t)}{Y(t)} - \frac{d\langle X,Y\rangle_t}{X(t)Y(t)}\frac{X(t)}{Y(t)}.
\end{equation}
Writing this in terms of $Z(t)$ gives
\begin{eqnarray}\nonumber
dZ(t) = Z(t)(\mu dt + \sigma  dW_t - \kappa (Z(t) - 1)dt - \eta d\tilde{W}(t) +  \frac{\eta^2Y^2(t)}{Y^2(t)}dt - \sigma\eta\rho dt)&& \\ 
= [\mu + \eta^2 - \sigma\eta\rho - \kappa(Z(t) - 1)]Z(t)dt + (\sigma dW(t) -\eta d\tilde{W}(t))Z(t),&&
\end{eqnarray}
which proves the lemma.
\end{proof}

For each control process $\pi \in \mathcal{U}_0$ we rewrite the dynamics of two-dimensional state process, $P=(V^\pi,Z)$, as follows
\begin{equation}
dP(t) = a(t,P(t),\pi(t))dt + b(t,P(t),\pi(t))dB(t). \label{eq:TwoDstateprocess}
\end{equation}
with initial value of $P(t_0) = p_0$ and $B = (W,\tilde{W})$ being the two-dimensional Brownian motion. The process $P$ is called the controlled process.
Let $[t_0,T]$ with $0\leq t_0<T<\infty$ be the relevant time interval and define $Q := [t_0,T)\times \mathbb{R}^2$. The coefficient functions 
\begin{eqnarray}
a \ \ &:& \ \ Q \times U \rightarrow \mathbb{R}^2, \\ \nonumber
b \ \ &:& \ \ Q \times U \rightarrow \mathbb{R}^{2\times 2},
\end{eqnarray}
are all continuous. Further, for all $\pi \in U$ let $a(\cdot, \cdot, \pi)$ and $b(\cdot, \cdot, \pi)$ be in $C^1(Q)$. We then define
\begin{definition}[Admissible control]
Denoting $\mathcal{A}^2$ the set of all admissible controls, we say a control $\left\lbrace \pi(t)\right\rbrace_{t\in[t_0,T]}$ will be called \textit{admissible} if the following conditions hold
\begin{itemize}
\item[(i)] $\forall k \in \mathbb{N}$ the integrability condition
\begin{equation}
E\left( \int_{t_0}^{T} |\pi(s)|^kds \right) < \infty
\end{equation}
is satisfied,
\item[(ii)] the corresponding state process $P^\pi$ satisfies
\begin{equation*}
E^{t_0,p_0}\left( \sup_{t \in [t_0,T]}|P^\pi(t)|^k\right) < \infty,
\end{equation*}
\item[(iii)] only pairs trading is allowed: short one of the asset and long the other
\begin{equation}
\pi_1 = -\pi_2. \label{eq:pairsTrading}
\end{equation} 
\end{itemize} 
\end{definition}
Since we consider self-financing portfolio, then by equation \eqref{eq:pairsTrading} the dynamics of the state process, $P = (V^\pi, Z)$, becomes
\begin{eqnarray*}
dV^\pi(t)=V^\pi(t)\left[ (\pi_1[\mu - \kappa(Z(t) -1)])dt  + \pi_1[\sigma dW(t)) + \eta d\tilde{W}(t)] \right], \ \ \ \ \ V^\pi(0)=v_0,&& \\ 
dZ(t)=[\mu +\eta^2 -\sigma\eta\rho -\kappa(Z(t) -1)]Z(t)dt + [\sigma dW(t) - \eta d\tilde{W}(t)]Z(t), \ \ Z(0)=z_0.&&
\end{eqnarray*}

\subsubsection*{Optimal investment strategy}
We assume that an investor's preference is represented by the power utility function
\begin{equation}\label{eq:powerUtility}
    U(x) = \frac{1}{\gamma} x^\gamma ,
\end{equation}
with $x \geq 0$ and risk aversion parameter $\gamma < 1$. Our aim is to maximize the objective functional $J$ over all admissible controls, i.e. determine an admissible control $\pi(\cdot)$ such that for each initial value $(t_0, v_0)$ the utility functional below is maximized:
\begin{equation}
J(t_0,v_0, z_0;\pi) := \mathbb{E}\left[U(V^\pi(T)) | V_{t_0} = v_0, Z_{t_0}=z_0 \right].
\end{equation}
The optimization problem is to find $\tilde{v}(t,v,z)$ and $\pi \in \mathcal{A}^2$ such that
\begin{equation}
\tilde{v}(t,v,z) := \sup_{\pi(\cdot) \in \mathcal{A}^2} J(t,v,z,\pi) = J(t,v,z,\pi^*). \label{eq:SCP}
\end{equation}
Consider the function $G(t,v,z)$ such that $G \in C^{1,2}(Q)$. The Hamilton-Jacobi-Bellman (HJB) equation corresponding to the stochastic control problem \eqref{eq:SCP} is 
\begin{eqnarray} 
\frac{\partial G}{\partial t}(t,v,z) + \sup_{\pi\in\mathcal{A}^2}\mathcal{L}^{\pi}G(t,v,z) = 0, \label{eq:HJB_original}
\end{eqnarray}
subject to terminal condition 
\begin{equation}
G(T,v,z) = v^\gamma.
\end{equation}
The infinitesimal generator, $\mathcal{L}^\pi G(t,v,z)$ in \eqref{eq:HJB_original} associated with the two dimensional state process $P = (V,Z)$ is given by
\begin{eqnarray}\nonumber
\mathcal{L}^\pi G(t,v,z) &=& \frac{1}{2}[\pi^2_1(\sigma^2 - 2\sigma \eta \rho +\eta^2) v^2 G_{vv} + 2\pi_1 (\sigma^2 - 2\sigma \eta \rho +\eta^2) vzG_{vz} +  \\ \nonumber
&& (\sigma^2 - 2\sigma\eta\rho  + \eta^2)z^2G_{zz}] + [\pi_1[\mu - \kappa(z-1)]]vG_v + \\
&&[ \mu +\eta^2 -\sigma\eta\rho - \kappa(z - 1)]zG_z. \label{eq:newinfinitesimalGenerator} 
\end{eqnarray}
\begin{theorem}
If there exists an optimal control $\pi^*(\cdot)$ then $G$ coincides with the value function:
\begin{equation*}
G(t,v,s) = \tilde{v}(t,v,z) = J(t,v;\pi^*).
\end{equation*}
\end{theorem}
Using separation ansatz we reduce a 3-dimensional HJB equation in \eqref{eq:HJB_original} to the following 2-dimensional PDE:
\begin{eqnarray}\nonumber
&&\tilde{\sigma}(\gamma-1) ff_t  - \frac{1}{2}\tilde{\sigma}^2\gamma z^2f_z^2 - \frac{1}{2}\gamma[\mu -\kappa(z-1)]^2f + \frac{1}{2}\tilde{\sigma}(\gamma-1)z^2ff_{zz} - \\ \nonumber
&& \tilde{\sigma}\gamma [\mu -\kappa(z-1)]zff_z + \tilde{\sigma}(\gamma -1)[\mu + \eta^2 -\sigma\eta\rho -\kappa(z-1)]ff_z,\\
&& \ \ \ \ \ \ \ \ \ \ \ \ \ \ \ \ \ \ \ \ \ \ \ \ \text{with} \ \ \ f(T,z) = 1,\ \ \ (t,z)\in [0,T]\times \mathbb{R}, \ \ \ \forall z \in \mathbb{R}, \label{eq:2dimPDE}
\end{eqnarray}
where $\tilde{\sigma} =  \sigma^2 - 2\sigma\eta\rho +\eta^2$. 

The issue at this stage is that this PDE does not have a closed for solution. This is a non standard PDE, which is not high dimensional but is nonlinear which makes using finite difference methods or any standard numerical methods inadequate. For this reason we propose to use the "Deep Galerkin Method" to  solve the PDE in \eqref{eq:2dimPDE}. Once the solution is found, we can write the optimal strategy as
\begin{eqnarray}\nonumber
 \pi_1^* &=& - \frac{\tilde{\sigma}zG_{vz} + [\mu - \kappa(z-1)]G_v}{\tilde{\sigma}vG_{vv}} 
 = - \frac{\tilde{\sigma}z(f_zv^{\gamma -1}\gamma) + [\mu - \kappa(z-1)](fv^{\gamma -1}\gamma)}{\tilde{\sigma}v(fv^{\gamma -2}\gamma (\gamma -1))} \\ 
 &=& - \frac{\tilde{\sigma}zf_z + [\mu - \kappa(z-1)]f}{\tilde{\sigma}f (\gamma -1)} = - \frac{zf_z}{(\gamma -1)f} - \frac{[\mu - \kappa(z-1)]}{\tilde{\sigma}(\gamma-1)}.
\end{eqnarray}
See Appendix \ref{appendix:HJB} for the details.

\subsection{Deep learning for solving PDE in stochastic control \label{sec:hybridMethod}}
Without an analytical solution to the non-standard 2-dimensional PDE in \eqref{eq:2dimPDE}, we approximate the solution with an algorithm "Deep Galekin Method" (DGM) proposed in \cite{DGM}. DGM is a merger of the Galerkin method and deep neural network machine learning algorithm. The Galerkin method is a popular numerical method which seeks a reduced-form solution to a PDE as a linear combination of basis functions. The deep learning algorithm, or DGM, uses a deep neural network instead of a linear combination of basis functions. The algorithm is trained on batches of randomly sampled time and  space points, therefore it is mesh free.
\subsubsection*{Brief review of DGM}
In general case, consider a PDE with $d$ spatial dimensions:
\begin{eqnarray}\nonumber
&& \frac{\partial u}{\partial t}(t,x;\theta) + \mathcal{L}u(t,x)= 0, \ \ \ (t,x) \in [0,T]\times \Omega , \\ \nonumber
&& u(t,x) = g(t,x), \ \ \ x \in \partial \Omega , \\
&& u(t=0,x) = u_0(x), \ \ \ x \in \Omega \label{eq:generalPDE}
\end{eqnarray}
where $x \in \Omega \subset \mathbb{R}^d$ and $\mathcal{L}$ is an operator of all the other partial derivatives. The goal is to approximate the $U(t,x)$ with deep neural network $f(t,x;\theta)$. Here $\theta\in\mathbb{R}^K$ are the neural network parameters. We want to minimize the objective function associated to the problem \eqref{eq:generalPDE} which consists of three parts:
\begin{enumerate}
\item[1.]  A measure of how well the approximation satisfies the PDE:
\begin{equation}
\left\|\frac{\partial f}{\partial t}(t,x;\theta) - \mathcal{L}f(t,x;\theta)\right\|^2_{[0,T]\times \Omega, \nu_1}.
\end{equation}
\item[2.] A measure of how well the approximation satisfies the boundary condition:
\begin{equation}
\left\|\frac{\partial f}{\partial t}(t,x;\theta) - g(t,x)\right\|^2_{[0,T]\times \partial \Omega, \nu_2}.
\end{equation}
\item[3.] A measure of how well the approximation satisfies the initial condition:
\begin{equation}
\left\|\frac{\partial f}{\partial t}(0,x;\theta) - u(0,x)\right\|^2_{\Omega, \nu_3}.
\end{equation}
\end{enumerate}
Here all three errors are measured in terms of $L^2$-norm, i.e. $\left\| f(y)\right\|^2_{\mathcal{Y},\nu} = \int_{\mathcal{Y}}|f(y)|^2\nu(y)dy$ with $\nu(y)$ being a density on region $\mathcal{Y}$. 

The sum of all three terms above gives us the objective function associated with the training of the neural network: 
\begin{eqnarray}\nonumber
J(f) &=& \left\|\frac{\partial f}{\partial t}(t,x;\theta) - \mathcal{L}f(t,x;\theta)\right\|^2_{[0,T]\times \Omega, \nu_1} + \\
&& \left\|\frac{\partial f}{\partial t}(t,x;\theta) - g(t,x)\right\|^2_{[0,T]\times \partial \Omega, \nu_2} + \left\|\frac{\partial f}{\partial t}(0,x;\theta) - u(0,x)\right\|^2_{\Omega, \nu_3}.
\end{eqnarray}
Thus, the goal is to find a set of parameters $\theta$ such that the function $f(t,x;\theta)$ minimizes the error $J(f)$. When the dimension $d$ is large, estimating $\theta$ by directly minimizing $J(f)$ is infeasible. Therefore, one can minimize the error $J(f)$ using a machine learning approach: stochastic gradient descent, where we use a sequence of time and space points drawn randomly.
The algorithm for DGM method is described in Algorithm \ref{algo:DGM} below.
\algsetup{indent=2em}
\newcommand{\DGM}{\ensuremath{\mbox{\sc Deep Galerkin Method}}}
\begin{algorithm}[h!]
\caption{\DGM()}
\begin{algorithmic}[1]
\label{algo:DGM}
\REQUIRE $\mathcal{L}f(), u(), g()$
\ENSURE $L^1_n + L^2_n + L^3_n$ is minimized
\medskip
\textbf{\\ Generate random points}: \\
\STATE $(t_n,x_n) \gets \mathcal{U}\sim\left[0,1\right]^2$\\
\STATE $(\tau_n,z_n) \gets \mathcal{U}\sim\left[0,1\right]^2$\\
\STATE $w_n \gets \mathcal{U}\sim\left[0,1\right]$\\
\STATE $s_n \gets ((t_n,x_n), (\tau_n,z_n), w_n)$\\

\textbf{\\ Calculate the squared error}: \\
\STATE $L^1_n  \gets \left(\frac{\partial f}{\partial t}(t_n,x_n;\theta_n) - \mathcal{L}f(t_n,x_n;\theta_n)\right)^2$

\STATE $L^2_n  \gets 
 \left(\frac{\partial f}{\partial t}(\tau_n,z_n;\theta_n) - g(\tau_n,z_n)\right)^2$

\STATE $L^3_n  \gets \left(\frac{\partial f}{\partial t}(0,x_n;\theta_n) - u(0,w_n)\right)^2$

\STATE $G(\theta_n, s_n)  \gets L^1_n + L^2_n + L^3_n$

\textbf{\\ Take a descent step at the random points}: \\
\STATE $- \underset{{\theta_n}}{\mathrm{argmax}} \ G(\theta_n, s_n)$
\STATE $\alpha_n \gets \alpha_{n-1}*\lambda$ 
\STATE $\theta_{n+1} \gets \theta_{n} - \alpha_{n} \nabla_{\theta}G(\theta_n, s_n)$\\

\textbf{Repeat until tolerance level $10^{-8}$ for convergence criterion  is achieved}
\end{algorithmic}
\end{algorithm}
\begin{remark}
The learning rate, $\alpha_n$, is a configurable hyperparameter\footnote{In machine learning, a hyperparameter is a parameter whose value is set before the learning process begins whereas, the values of other parameters are derived via training.} used in the training of neural networks that controls how much to change the model in response to the estimated error. Each time the model weights are updated. Learning rate has a small positive value, often in the range between $0.0$ and $1.0$. Similar to \cite{SolvingPDEwithDL}, we set $\alpha_0 = 0.001$. Note that our learning rate $\alpha_n$ must decrease with $n$, see \cite{DGM}, and a simple enough way to do that is by using an exponential weighted method where $\alpha_n \gets \alpha_{n-1}*\lambda$ with $\lambda \in \left]0,1\right[$.
\end{remark}
The neural network (NN) architecture used in DGM is like a long short-term networks (LSTMs) though with small differences, see \cite{DGM}. We describe below the architecture of this NN:
\begin{equation*}
\begin{subequations}
\begin{aligned}
&S^1 = \sigma(\mathbf{w}^1\cdot \mathbf{x}+b^1) \\
&Z^l = \sigma(\mathbf{u}^{z,l}\cdot \mathbf{x}+\mathbf{w}^{z,l}\cdot S^l+b^{z,l}) & l = 1, \ldots, L\\
&G^l = \sigma(\mathbf{u}^{g,l}\cdot \mathbf{x}+\mathbf{w}^{g,l}\cdot S^l+b^{g,l}) & l = 1, \ldots, L \\
&R^l = \sigma(\mathbf{u}^{r,l}\cdot \mathbf{x}+\mathbf{w}^{r,l}\cdot S^l+b^{r,l}) & l = 1, \ldots, L\\
&H^l = \sigma(\mathbf{u}^{h,l}\cdot \mathbf{x}+\mathbf{w}^{h,l}\cdot (S^l\odot R^l)+b^{h,l}) & l = 1, \ldots, L \\
&S^{l+1} = (1-G^l)\odot H^l + Z^l \odot S^l & l = 1, \ldots, L\\
&f(t,\mathbf{x},\mathbf{\theta}) = \mathbf{w} \cdot S^{L+1} + \mathbf{b}
\end{aligned}
\end{subequations}
\end{equation*}
with $\odot$ denoting Hadamard multiplication, L number of layers and $\sigma$ the activation function. The rest of the subscript refer to the neurones for our NN architecture of Figures \ref{fig:birdEyeDGM} and \ref{fig:detailedDGM}.
\begin{remark}
We can see the Bird Eye view of the DGM \cite{SolvingPDEwithDL, DGM} method in Figure \ref{fig:birdEyeDGM} and its details in Figure \ref{fig:detailedDGM}. The rationale is explained in \cite{SolvingPDEwithDL, DGM}.
\end{remark}

\subsubsection*{Testing DGM on Merton problem}
The method was tested with several nonlinear, high-dimensional PDEs independently in \cite{SolvingPDEwithDL} and \cite{DGM}, including nonlinear HJB equations. We have tested the DGM algorithm on HJB equation for the Merton problem ourselves. More specifically, Figures \ref{fig:mertonNNAnalitycal} and \ref{fig:mertonNN} show the plots of the analytical and approximated surface with DGM solution. Figure \ref{fig:mertonerror} shows the difference between analytical and approximate solution. The approximation is good. Most of the time, the error is between $0\%$ and $1\%$. The approximate solution does not do as well around $t=0$ (the maximum error of $4\%$ is around $t=0$). This corroborates with the findings in \cite{SolvingPDEwithDL}.


\subsubsection*{Solution to our PDE problem using DGM}\label{sec:DGMforOurpde}
Recall the PDE we want to solve is given in equation \eqref{eq:2dimPDE}. In the absence of a closed form solution to this PDE we approximate the solution with the DGM algorithm described above. Figure \ref{fig:changeInRhoMu} shows the approximate solution to the PDE in \eqref{eq:2dimPDE} for different parameter values. Recall, once we have the numerical solution $f$ for the PDE above, we obtain the optimal weights as following:
\begin{eqnarray}\nonumber
 \pi_1^* &=& - \frac{\tilde{\sigma}zG_{vz} + [\mu - \kappa(z-1)]G_v}{\tilde{\sigma}vG_{vv}} = - \frac{\tilde{\sigma}z(f_zv^{\gamma -1}\gamma) - [\mu - \kappa(z-1)](fv^{\gamma -1}\gamma)}{\tilde{\sigma}v(fv^{\gamma -2}\gamma (\gamma -1))} \\ 
 &=& - \frac{\tilde{\sigma}zf_z - [\mu - \kappa(z-1)]f}{\tilde{\sigma}f (\gamma -1)} = - \frac{zf_z}{(\gamma -1)f} - \frac{[\mu - \kappa(z-1)]}{\tilde{\sigma}(\gamma-1)}, \label{eq:optimalWeightSC}
\end{eqnarray}
with $\pi_1^* = - \pi_2^*$.

\subsection{Dynamic Switching between optimal strategies of mean-variance and power utility}\label{sec:DynamicSwitching}
Although in the previous two cases we assume that an investor has a certain risk preferences as modelled by a utility function (MVC and power utility), it is interesting to consider a limiting case where the investor can be always persuaded to go for more money (identical utility function $U(x) = x$, which is essentially the power utility function with risk aversion parameter $\gamma =1$) when deciding between MVC or power utility.

Assuming that an investors' preference  is modelled either as in equation \eqref{eq:utilityMVC} or as in equation \eqref{eq:powerUtility}, in order to improve further the portfolio returns we employ dynamic switching between the two optimal strategies
\begin{eqnarray}
 \psi^*(t) = \begin{cases} \pi^*(t), \ \ \text{if} \ \ V^{\pi^*}(t) \geq V^{h^*}(t), \\
h^*(t), \ \ \text{otherwise},
 \end{cases}
\end{eqnarray}
where $\pi^*(t)$ and $h^*(t)$ are given in equations \eqref{eq:optimalWeightSC} and \eqref{eq:optimalweighMVC} and $V^{\pi^*}(t)$ and $V^{h^*}(t)$ are given in equations \eqref{eq:portvalpi} and \eqref{eq:portvalh}. The motivation behind the dynamic switching is that the investor wants to benefit from both the mean-reversion and the correlation elements of the cointelation model \eqref{eq:modelCointel}. More specifically, as the spread between two assets increases the investor implements pairs trading and makes profit, otherwise the MVC approach is used. 

The portfolio return over investment horizon $[0,T]$ with $T=1000$ days is
\begin{equation}
    R(r_p) = \frac{V(0) - V(T)}{V(0)}.
\end{equation}
We perform 500 simulations with the same model and present in Table \ref{tab:Table_Averag} the average results. The average return at terminal time $T$ obtained by using dynamic switching optimal strategies is higher than the average returns calculated by employing MVC or power utility maximizing optimal strategies.

\section{Machine Learning formulation of the portfolio optimization problem}\label{sec:ML_approach}
\subsection{The portfolio optimization problem}
We assume an initial wealth $\tilde{w}_0 > 0$ at time $t=0$. The investment behaviour is modeled by an investment strategy $w = (w_1, w_2)$. Here $w_1(t)$, $w_2(t)$ denote the percentages of wealth invested in asset $X$ and $Y$ respectively at time $t$. Let $V(t)$ denote the portfolio value at time $t$ and $V^{PnL}(t) := V(t) - V(0)$ denote the profit ant loss (P\&L) over $[0,t]$. At each time $t$ we allow either pairs trading: $w_1(t) = -w_2(t)$ or long only strategies without leverage: $w_1(t) + w_2(t) = 1$ with $w_1(t), w_2(t) > 0$.

The general optimization problem is to find an optimal strategy, $w(t)$, such that the
terminal P\&L is maximized:  
\begin{equation}
    w^*(t) = \underset{w(t) \in \mathcal{A}}{\mathrm{argmax}} \ V^{PnL}(w,T), 
\end{equation}
where $V^{PnL}(w,T)$ is profit and loss corresponding to the strategy $w$ at time terminal time $T$. We use clustering analysis to device the bands and in each band we solve the following optimization problem
\begin{equation}
    w_i^*(t) = \underset{w_i(t) \in \mathcal{A}}{\mathrm{argmax}} \ V^{PnL}(w_i,t),
\end{equation}
where $i=1,...,n$ is the number of bands, $V^{PnL}(w_i,t)$ is  profit and loss corresponding to the strategy $w_i$ at time $t$. Then the overall solution $w^*$ is obtained via a linear interpolation of optimal weights per band $w_i^*$

The advantage of the proposed method is that we do not impose certain model on the asset prices. Only data observations are required to calculate the optimal weight, meaning that the complex SDE calibration is avoided.

\subsection{Review of Band-Wise Gaussian Mixture model \label{sec:bandwiseGaussianMix}}
We review band-wise Gaussian mixture model because it inspires our method of selecting the bands. Consider a probability space $(\Omega, \mathcal{F}, \mathbb{P})$ and let $(P_t)_{t \geq 0}$ denote the asset price. Mahdavi-Damghani and Roberts \cite{bandwiseGaussianMix} has recently introduced a generalised bumping SDE for the price dynamics of asset $P_t$. The SDE contains some secondary parameters whose purpose is empirical manual fitting. The generalized SDE is given by
\begin{equation}
dP_t= \theta_{t,\tau} (\mu_{t,\tau} - P_t) dt + \sigma P_t^{\alpha}(1-P_t^2)^{\beta} dW_t.
\label{eq:generalisedBM}
\end{equation}
Here $\theta_t$ is the speed of mean reversion, $\mu_t$ is the long term mean, $\alpha$ is the positivity flag enforcer, $\beta$ is the $[-1,+1]$ boundary flag enforcer and $\{\bigcup dW_i\}_{i=t-\tau}^{t}$ is the set of historical deviations of the assumed model's distribution (e.g.: all the historical absolute returns in the context of a normal diffusion).

This generalised SDE gives as a special case the cointelation model: take $\theta = -\mu$, $\mu = 0$, $\alpha = 1$ and $\beta = 0$ for the dynamic of $X$ in \eqref{eq:modelCointel}; take $\theta_{t, \tau} = \kappa$, $\mu_{t,\tau} = X_t$, $\alpha = 1$, $\beta = 0$ for the dynamics of $Y$ in \eqref{eq:modelCointel}. The SDE in \eqref{eq:generalisedBM} can also model:
\begin{itemize}
	\item Proportional returns (log-normal diffusion) when $\theta=0$, $\alpha=1$, $\beta=0$.
	\item Absolute returns (normal diffusion) when $\theta=0$, $\alpha=0$, $\beta=0$,
	\item Mean reverting returns where we enforce positivity of returns (e.g. CIR \cite{CIR} diffusion when $\alpha = 1/2$ and $\beta = 0$),
	\item Mean reverting returns where we do not enforce positivity of the returns (e.g OU \cite{OU} diffusion when $\alpha = 0$ and $\beta = 0$).
\end{itemize}
In general calibrating parameters of the SDE in  \eqref{eq:generalisedBM} to a real data is complex. Using data simulated with \eqref{eq:generalisedBM} their empirical distribution is approximated for the purpose of prediction by a band-wise Gaussian mixture model. This is done for a sequence of bands which are created using Machine Learning clustering method (see \cite{bandwiseGaussianMix}).

Let $\textbf{P}=\{p_1, \ldots, p_n\}$ be a set of empirical random variables sampled using equation \eqref{eq:generalisedBM} with cumulative distribution function $F(p)$ and density $f(p)$. Denote $O=\{p_{(1)}, \ldots, p_{(n)}\}$ the ordered set of $\textbf{P}$ such that $p_{(1)} < p_{(2)} < \ldots  < p_{(n)}$ and 
\begin{equation*}
    O^{i}_h=\{p_{(\ceil{n((i-1)+1)/h})}, \ldots, p_{(\floor{n(i)/h})}\}.
\end{equation*}
Then the \textit{band-wise Gaussian mixture model} for the empirical distribution function of the data simulated using the SDE in equation \eqref{eq:generalisedBM} is given as follows:
\begin{equation}
{\hat  F}_{n}(p_{i}|\mathcal{F}_{t})={\frac  {1}{n}}\sum _{{j=1}}^{h}\sum _{{i=\eta}}^{\zeta}{\mathbf  {1}}_{{p_{i}\in O^{j}_h}}
\label{eq:EmpiricalGaussianMixture}
\end{equation}
with $\eta = \ceil{n((i-1)+1)/h}$ and $\zeta=\floor{n(i)/h}$.

For example in the case bands $h=3$, using a Gaussian Mixture such that
\begin{equation}
    {\hat  F}_{n}(p_{i}|\mathcal{F}_{t})=\mathcal{N}(-3,1){\mathbf  {1}}_{{p_{t}\in O^{1}_3}}+\mathcal{N}(0,1){\mathbf  {1}}_{{p_{t}\in O^{2}_3}}+\mathcal{N}(3,1){\mathbf  {1}}_{{p_{t}\in O^{3}_3}},
\end{equation}
we obtain the approximate stratification in Figure \ref{fig:GaussianMixture}. The stratification is made so that the cardinality in each $O^{j}_h$ region remains approximately the same, as opposed to being the result of a geometrical separation function of $p_{(1)}$ and $p_{(n)}$.

Theorem 1 in \cite{bandwiseGaussianMix} ensures a good approximation of the generalised SDE \eqref{eq:generalisedBM} by the Gaussian mixture model \eqref{eq:EmpiricalGaussianMixture}. The calibration for the band-wise Gaussian mixture is given in Algorithm 
\ref{algo:pGaussianMixture}.

For our optimization problem we take a similar approach of dividing the range of observations into bands via the clustering algorithm, and then perform an optimization in each band via perturbation of weights.
\algsetup{indent=2em}
\newcommand{\BWGM}{\ensuremath{\mbox{\sc Band-Wise Gaussian Mixture}}}
\begin{algorithm}[h!]
\caption{\BWGM($P, h$)}
\begin{algorithmic}[1]
\label{algo:pGaussianMixture}
\REQUIRE array $P_{1:n}$ and number of bands $h$
\ENSURE $\Omega^{(1:h)}$, $[B^{+}_{(1:h)}, B^{-}_{(1:h)}]$ are returned.
\medskip
\textbf{\\ Sorting state}: \\
\STATE $X_{(1:h)}$ $\gets$ QuickSort($X_{1:n}$)\\
\STATE $[B^{+}_{(1:h)}, B^{-}_{(1:h)}]$ $\gets$ FindPercentileBands($X_{(1:n)}$, $h$)\\
\STATE $\Omega^{(1:\ceil{n/h})}\gets[]$\\
\textbf{\\ Allocation state}:\\

\FOR{$j=1$ to $h$}
\FOR{$i=1$ to $n$}
\IF{$B^{-}_{(1:h)}\leq P_{(i)}<B^{+}_{(1:h)}$}
\STATE
Amend($\Omega^{(j)}, P_{(i)}$)
\ENDIF
\ENDFOR
\ENDFOR

\textbf{\\ Checking Approximation state}:\\
\STATE
$\hat{\mu}_{1:h} \gets$ mean($\Omega^{(1:h)}$)\\
\STATE
$\hat{\sigma}_{1:h} \gets$ stdev($\Omega^{(1:h)}$)\\
\STATE
Print($\cup_{i=1}^{h} \mathcal{N}(\hat{\mu}_i,\hat{\sigma}_i)$)\\

\textbf{\\ Return state}:\\
\STATE $\Omega^{(1:h)}$, $[B^{+}_{(1:h)}, B^{-}_{(1:h)}]$

\end{algorithmic}
\end{algorithm}

\subsection{Optimal Machine Learning strategy}\label{sec:sigDef}
Based on the idea of band-wise Gaussian mixture model, we use clustering analysis to create bands, however not for the observed asset price data, but for the spread between two asset prices in \eqref{eq:modelCointel}, i.e. $X_t - Y_t$. Inside of each band instead of specifying the distribution as in band-wise Gaussian mixture, we test a set of strategies that maximizes the corresponding P\&L. We record the optimal strategies within each band, and in live trading, whenever the spread of asset prices falls in a certain band we employ the optimal strategy for this specific band.

We now present the trading signal that translates to investment strategy in machine learning approach.
\paragraph*{The Bayesian set-up:}
We set from equation \eqref{eq:modelCointel} $B_t = X_{t} - Y_{t}$ and have
\begin{equation*}
B_t = \{B^{+}_{n,t}, B^{+}_{n-1,t}, \ldots, B^{+}_{1,t}, B^{1}_{1,t}, \ldots, B^{-}_{n-1,t}, B^{-}_{n,t}\}, 
\end{equation*}
such that $B^{+}_{n,t} > B^{+}_{n-1,t} > \ldots > B^{+}_{1,t} > 0 > B^{-}_{1,t} > \ldots > B^{-}_{n-1,t} > B^{-}_{n,t}$.
We know that depending on the spread, the resulting approximated distribution of the samples differ \cite{bandwiseGaussianMix}. The calibration algorithm will then consist of creating as many zones as possible whilst and as many strategies as possible within these bands and test how well each strategy is doing in each band in terms of P\&L maximization. We take a \textit{direct} approach (see Remark \ref{rem:direct}) consisting of 3 strategies and their cumulative P\&Ls. Fixing the bands $[a_i, b_i]$, with $i =1, 2, ...,n$ we consider the following strategies:
\begin{itemize}
\item Strategy $S^{++}$ in which we are long both $X$ and $Y$ at time $t$ in between bands $[a_i,b_i]$ and with P\&L $V^{++}_{[a_i,b_i],t}$.
\smallskip
\item Strategy $S^{+-}$ in which we are long $X$ and short $Y$ at time $t$ in between bands $[a_i,b_i]$ and with P\&L $V^{+-}_{[a_i,b_i],t}$.
\smallskip
\item  Strategy $S^{-+}$ in which we are short $X$ and long $Y$ at time $t$ in between bands $[a_i,b_i]$ and with P\&L $V^{-+}_{[a_i,b_i],t}$.
\end{itemize}
The P\&Ls corresponding to these strategies are defined as following:
\begin{eqnarray}\nonumber
V^{++}_{[a_i,b_i],T}  &=& \sum_{t=0}^T [w^{++}_{[a_i,b_i],t}\Delta X_t + (1-w^{++}_{[a_i,b_i],t}) \Delta Y_{t}] 1_{a_i< \Delta_t\leq b_i}, \nonumber \\
V^{+-}_{[a_i,b_i],T}  &=& \sum_{t=0}^T [w^{+-}_{[a_i,b_i],t}\Delta X_t - (1-w^{+-}_{[a_i,b_i],t}) \Delta Y_{t}] 1_{a_i< \Delta_t\leq b_i}, \nonumber \\
V^{-+}_{[a_i,b_i],T}  &=& \sum_{t=0}^T [-w^{-+}_{[a_i,b_i],t}\Delta X_t + (1-w^{-+}_{[a_i,b_i],t}) \Delta Y_{t}] 1_{a_i< \Delta_t\leq b_i}. \nonumber
\end{eqnarray}
\begin{remark}\label{rem:direct}
We call this approach \textit{direct}, since ideally the number of strategies should consist of a more granular weight distribution. However for the sake of comparing with Financial Mathematics approach we consider the same set of strategies: long only, long/short.
\end{remark}
We denote the maximum P\&L achieved by each of these strategies by $V^{\mp\mp,*}_{[a_i,b_i],T}$, as given by equation \eqref{eq:maxPnLInEachBand} and define $S^{**}_{[a_i,b_i],T}$ of P\&L $V^{**}_{[a_i,b_i],T}$ (equation \eqref{eq:pDoubleStar}), the optimal strategy using Gaussian Learning in band $[a_i,b_i]$. 
\begin{equation}
\label{eq:maxPnLInEachBand}
V^{\mp\mp,*}_{[a_i,b_i],T} = \underset{w^{\mp\mp}_{[a_i,b_i],t\in[0,T]}}{\mathrm{argmax}} V^{\mp}_{[a_i,b_i],T}, \ w^{\mp\mp}_{[a_i,b_i],t} \in [0,1] 
\end{equation}
\begin{equation}
\label{eq:pDoubleStar}
V^{**}_{[a_i,b_i],T} = \max (V^{++,*}_{[a_i,b_i],T}, V^{+-,*}_{[a_i,b_i],T}, V^{-+,*}_{[a_i,b_i],T}).
\end{equation}
In live trading we recombine the optimal weights per bands into an overall optimal solution via a linear interpolation:
\begin{equation}
    w^*(t) = \sum_{i=1}^n w^*_i \mathbf{1}_{\{(X_t - Y_t) \in [a_i,b_i]\}}.\label{eq:interpolStrategy}
\end{equation}
Although we do not have a proof that the resulting interpolated strategy in \eqref{eq:interpolStrategy} is optimal, we use it as a benchmark that still improves over the results with Financial Mathematics approach. Our goal is to apply Machine Learning approach to a pair of assets that exhibit some dependence, but this approach can be used for any model, i.e. it is model agnostic.
\algsetup{indent=2em}
\newcommand{\BWGMforCoint}{\ensuremath{\mbox{\sc Band-Wise ML for Cointelation}}}
\begin{algorithm}[h!]
\caption{\BWGMforCoint($P, h$)}
\begin{algorithmic}[1]
\label{algo:pGaussianMixtureForCoint}
\REQUIRE array $P_{1:n}$ and number of bands $h$
\ENSURE $\Omega^{(1:h)}$, $[B^{+}_{(1:h)}, B^{-}_{(1:h)}]$ are returned
\medskip
\textbf{\\ Sorting state}: \\
\STATE $P_{(1:h)}$ $\gets$ QuickSort($P_{1:h}$)\\
\STATE $[B^{+}_{(1:\frac{h}{2})}, B^{-}_{(1:\frac{h}{2})}]$ $\gets$ FindPercentileBands($P_{(1:n)}$, $h$)\\
\STATE $B_{(1:h)}$ $\gets$  $[B^{+}_{(1:\frac{h}{2})}, B^{-}_{(1:\frac{h}{2})}]$ \\
\STATE $\Omega^{(1:\ceil{n/h})}\gets[]$\\
\textbf{\\ Allocation state}:\\

\FOR{$j=1$ to $h$}
\FOR{$i=1$ to $n$}
\IF{$P_{(i)} \in B^{i}$}
\STATE
Amend($\Omega^{(j)}, P_{(i)}$)
\ENDIF
\ENDFOR
\ENDFOR 

\textbf{\\ Optimize the 3 types of P\&L for each band}:\\
\FOR{$i=1$ to $h$}
\STATE $V^{++,*}_{B_i, T} \gets \underset{w^{++}_{B_i, t\in[0,T]}}{\mathrm{argmax}}  V^{++}_{B_i,T}$ \\
\STATE $V^{+-,*}_{B_i, T} \gets\underset{w^{+-}_{B_i, t\in[0,T]}}{\mathrm{argmax}}  V^{+-}_{B_i,T}$ \\
\STATE $V^{-+,*}_{B_i, T} \gets \underset{w^{-+}_{B_i, t\in[0,T]}}{\mathrm{argmax}} V^{-+}_{B_i,T}$ \\
\ENDFOR

\textbf{\\ Rank and return best strategy for each band}:\\
\FOR{$i=1$ to $h$}
\STATE $V^{**}_{B_i,T} \gets \max (V^{++,*}_{B_i,T}, V^{+-,*}_{B_{i,T}}, V^{-+,*}_{B_{i,T}})$
\STATE $S^{*}_T \gets$ ($S^{++,*}_{B_i,T},S^{+-,*}_{B_i,T},S^{-+,*}_{B_i,T}$)\\
\STATE $S^{**}_{B_{i,T}}\gets$ returnCorrespondingStrat($V^{**}_{B_{i,T}},S^{*}_T$),
\ENDFOR

\textbf{\\ Forecasting }:\\
\STATE signal$^S$, signal$^{S_l}\gets$ forecast($S^{**}_{B_{i,T}}, S_t,  S_{l,t}$)

\textbf{\\ Return buy/sell signals}:\\
\STATE signal$^S$, signal$^S_l$
\end{algorithmic}
\end{algorithm}

We further provide Algorithm \ref{algo:pGaussianMixtureForCoint} as the pseudo-code for the calibration process. Note that in both Algorithms \ref{algo:pGaussianMixture} and \ref{algo:pGaussianMixtureForCoint}, we have used a QuickSort which can be substituted by other sorting algorithms. Note that the use of self explanatory functions such as \textit{returnCorrespondingStrat(x,y)} in line 20 of Algorithm \ref{algo:pGaussianMixtureForCoint} which given the set of strategies and the P\&L returns, as its name indicates, outputs the corresponding strategy that maximizes P\&L. The function \textit{forecast(x,y,z)} in line 22 of Algorithm \ref{algo:pGaussianMixtureForCoint} takes as input the set of trained strategies and the current level of $X_t$ and $Y_t$ and returns a prediction of where the signals for the latter two should be. 
Finally the use of the \textit{argmax} function in lines 13-16 can be replaced by a simple for loop but in the interest of not making the pseudocode too crowded we have kept it this way.
\begin{remark}
In \cite{bandwiseGaussianMix} authors show that a reasonable risk manager or trader can assume the generalized SDE \eqref{eq:generalisedBM} with $\beta=0$ and an $\alpha=1$, in order to enforce positivity for the simulated scenarios of our risk factor. This very reasonable assumption would have crashed the whole risk engine if it is no longer satisfied in the real markets. The approach we advocate would have, however, been able to continue its dynamical learning scenario without any problem since it is model agnostic. 
\end{remark}

\section{Numerical results \label{sec:resultsAndConclusion}}
Figure \ref{fig:MLvsDS_1path} illustrates the ML and the DS approaches on one single simulated path. Note that when implementing the ML approach with a horizon of 1000 days, we double this data for training, i.e. we use 2000 historical daily prices. We have performed two sets of 500 simulations and we have gathered their results in the following two examples. 

\begin{example} We have simulated 500 paths of $X$ and $Y$ based on cointelation model \eqref{eq:modelCointel} with parameters $\mu=0.05, \sigma=0.17, \eta=0.16, \kappa=0.1, \rho = -0.6$. Figure \ref{fig:SC_vs_ML} illustrates that the Machine Learning approach with long/short strategies ($ML_{LS}$), on average performs slightly better in terms of P\&L than the Stochastic Control approach (SC). However, based on histogram none of the approaches perform  significantly better or significantly worse than the other at any time. \label{ex:SCvsML}
\end{example}
\begin{example}
We have simulated 500 paths of $X$ and $Y$ based on cointelation model \eqref{eq:modelCointel} with parameters $\mu=0.05, \sigma=0.17, \eta=0.16, \kappa=0.1, \rho = -0.6$. Figure \ref{fig:DS_vs_ML} illustrates how the ML approach seems to perform slightly better in terms of P\&L than the FM approach about $55\%$ of the time, while being outperformed the other $45\%$ of the time. However, based on histogram we have noticed that sometimes the ML approach is being outperformed significantly more than it outperforms FM approach. 
 \label{ex:DSvsML}
\end{example}

From histogram of performance in Figures \ref{fig:SC_vs_ML} and \ref{fig:DS_vs_ML} we have concluded that for parameters  $\mu=0.05, \sigma=0.17, \eta=0.16, \kappa=0.1, \rho = -0.6$ of cointelation model \eqref{eq:modelCointel} we have the following rankings for the approaches:
\begin{equation*}
    SC < ML_{LS} < FM< ML. 
\end{equation*}
The reason for ML with full set of strategies (long only and long/short) outperforming the DS most of the time might be the fact that in long only optimal strategies of ML approach we have more variety in weights, whereas the closed form formula \eqref{eq:optimalweighMVC} in FM gives us almost constant weights (with small fluctuations).

\subsection*{Possible directions for future work}
\subsubsection*{Multidimensional case} 
One direction for future work is to consider portfolio optimization problem for n-dimensional cointelation model. For instance, when $n=3$ we can have something of the following form:
\begin{eqnarray} \nonumber
    dS_{t}^a &=& \sigma S_{t}^a dW_t^a \\
    dS_{t}^b &=& \theta (S_{t}^a - S_{t}^b) dt + \sigma S_{t}^a dW_t^b \\ \nonumber 
    dS_{t}^c &=& \theta (S_{t}^a - S_{t}^c) dt + \sigma S_{t}^c dW_t^c
	\label{eq:cointelation3a}
\end{eqnarray}
One natural question would first be about how to model this triplet? For instance would equation \eqref{eq:cointelation3a} with $S^b$ and $S^c$ reverting around $S^a$ be more in-line with the pair from equation \eqref{eq:modelCointel} or would $S^b$ reverting around $S^a$ and $S^c$ reverting around $S^b$ be better? Are they equivalent or is one more useful? What happens as $n$ increases? We plan to examine these questions in the future. 
\subsubsection*{Application to cryptocurrencies}
Another direction for future work is to model cryptocurrencies prices with the cointelation model and construct cryptocurrency indices using the portfolio optimization approaches proposed in this paper. Cryptocurrencies offer a source of alternative alpha, therefore there has been an emergence of cryptocurrency indices in the recent past with construction methodology ranging in the spectrum of Risk Parity to Stochastic Portfolio Theory (SPT). Given the spectacular volatility of the cryptocurrency market, even though the point of the index is to reduce the overall volatility, the index position remains fundamentally long. However, using a combination of beta neutral approach (long/short strategies) with an occasional long only alternative could be the winning combination for this asset class. For this reason the cointelation model is an interesting model to use for this application.

\subsection*{Conclusion}
We have studied the portfolio optimization problem of two assets that follow the cointelation model using two approaches: Financial Mathematics and Machine Learning. We first implemented the FM approach, where we use classic financial mathematics criteria: mean-variance and power utility maximization. Without an analytical solution to the PDE \eqref{eq:2dimPDE}, we resort to the DGM method, a deep learning algorithm, to solve it numerically.
The second approach we implemented is ML using clustering. The latter approach is easier to implement, it is model agnostic, therefore avoids the complex SDE calibration. In our case the Machine Learning approach slightly outperforms the Financial Mathematics approach.
\clearpage

\section*{Appendices}
\appendix
\section{Proof of Lemma \ref{lem:valueOfPortfolio} \label{appendix:MVC}}
Since $X_t$ is a geometric Brownian motion, we have 
\begin{equation}
E[r(X_t)] = (\mu - \frac{\sigma^2}{2})\Delta t 
\end{equation}
where $X_{t-\Delta t}$ is a known constant at time $t - \Delta t$.
The expectation of log return of asset $Y$ is
\begin{equation}
E[r(Y_t)] = E[\ln(Y_t)] - \ln(Y_{t-\Delta t}),
\end{equation}
where $Y_{t-\Delta t}$ is a known constant at time $t - \Delta t$. We use Taylor expansion to approximate expected value and variance of $\ln(Y_t)$ and covariance of $\ln(Y_t)$ and $\ln(X_t)$ (see \cite{Benaroya}, p.165-167):
\begin{eqnarray}\label{eq:expec_logY_approx}
E[\ln(Y_t)] &\approx& \ln\left(E[Y_t]\right) - \frac{\sigma^2[Y_t]}{2E[Y_t]^2}, \\ 
\sigma^2[\ln(Y_t)] &\approx& \frac{\sigma^2[Y_t]}{E[Y_t]^2} \label{eq:variance_logY_appox},\\ 
\sigma[\ln(Y_t)\ln(X_t)] &\approx& \ln\left(1+\frac{\sigma[X_tY_t]}{E[X_t]E[Y_t]}\right). \label{eq:cov_logXY_approx}
\end{eqnarray}
First, we need to derive $E[Y_t]$. From equation \eqref{eq:modelCointel} we have
\begin{eqnarray}
Y_t = Y_{t-\Delta t} + \kappa\int_{t-\Delta t}^t(X_s - Y_s)ds + \eta\int_{t-\Delta t}^tY_sdZ_s.
\end{eqnarray}
Taking expectation on both sides we have
\begin{equation}
E[Y_t] = Y_{t-\Delta t} + \kappa\int_{t-\Delta t}^tE[X_s - Y_s]ds.
\end{equation}
Differentiating on both sides we get
\begin{eqnarray}
\frac{dE[Y_t]}{dt} = \kappa E[X_t] - \kappa E[Y_t] = \kappa X_{t-\Delta t} e^{\mu \Delta t} -\kappa E[Y_t].
\end{eqnarray}
Denoting $E[Y_t]=y(t)$ we obtain an ordinary differential equation (ODE):
\begin{equation}
y' = -\kappa y + \kappa X_{t-\Delta t}e^{\mu \Delta t}.
\end{equation}
The solution is given by 
\begin{equation}
y(t) = E[Y_t] = ae^{\mu \Delta t} + (Y_{t-\Delta t} - a)e^{-\kappa \Delta t},
\end{equation}
where 
\begin{equation*}
  a = \frac{\kappa X_{t-\Delta t}}{\mu+\kappa}.   
\end{equation*}
In order to derive $E[Y_t^2]$ we first compute $E[X_tY_t]$. Applying integration by parts (IBP) to \eqref{eq:modelCointel} we get
\begin{eqnarray} \nonumber
&& d(X_tY_t) = X_tdY_t + Y_tdX_t + dX_tdY_t \\ 
&& = \kappa X_t^2dt - \kappa X_tY_tdt + \eta X_tY_tdW_t +  \mu X_tY_tdt +\sigma X_tY_tdZ_t + \sigma \eta \rho X_tY_tdt.
\end{eqnarray}
Thus 
\begin{eqnarray*}
X_tY_t = X_{t-\Delta t}Y_{t-\Delta t} + \kappa \int_{t-\Delta t}^tX_s^2ds + \eta\int_{t-\Delta t}^tX_sY_sdW_s + \\
(\mu - \kappa +\sigma\eta\rho)\int_{t-\Delta t}^t X_sY_sds + \sigma \int_{t-\Delta t}^t X_sY_sdZ_s.
\end{eqnarray*}
Taking expectation and differentiating on both sides 
\begin{eqnarray}
\frac{dE[X_tY_t]}{dt} = \kappa E[X_t^2] +(\mu - \kappa + \sigma\eta\rho)E[X_tY_t]. 
 \end{eqnarray}
Denoting $E[X_tY_t]=x(t)$ we obtain ODE
\begin{equation}\label{eq:ODEforXY}
x' = \kappa E[X_t^2] + (\mu + \sigma \eta \rho - \kappa)y.
\end{equation}
Since $X_t$ is GBM, its second moment is given by
\begin{eqnarray}
E[X_t^2] = E[X_{t-\Delta t}^2e^{(2\mu - \sigma^2)\Delta t + 2\sigma W_t}] = X_{t-\Delta t}^2e^{(2\mu+\sigma^2)\Delta t}. 
\end{eqnarray}
Thus \eqref{eq:ODEforXY} becomes
\begin{equation}
x' = \kappa X_{t-\Delta t}^2e^{(2\mu + \sigma^2)\Delta t} + (\mu - \kappa + \sigma \eta \rho )y.
\end{equation}
Using variation of parameters method we get the solution
\begin{eqnarray} \label{eq:solutionXY}
x(t) = E[X_tY_t] = be^{(2\mu+\sigma^2)\Delta t} + (X_{t-\Delta t}Y_{t-\Delta t} - b)e^{(\mu - \kappa + \sigma \eta \rho )\Delta t},
\end{eqnarray}
where 
\begin{equation*}
    b = \frac{\kappa X_{t-\Delta t}^2}{\mu+\sigma^2 + \kappa - \sigma\eta\rho}.
\end{equation*}
Now we are ready to compute $E[Y_t^2]$. By It\^o's lemma the dynamics of $Y_t^2$ is
\begin{eqnarray}
dY_t^2 = 2Y_tdY_t + (dY_t)^2 = (\eta^2 - 2\kappa)Y_t^2dt + 2\kappa X_tY_tdt + 2\eta Y_t^2dZ_t. 
\end{eqnarray}
Integrating on both sides
\begin{eqnarray}
Y_t^2 = Y_0^2 + 2(\eta^2 -\kappa)\int_0^tY_s^2ds + 2\kappa \int_0^t X_sY_sds +2\eta\int_0^tY_s^2dZ_s.  
\end{eqnarray}
Taking expectation on both sides and and differentiating 
\begin{equation}
\frac{dE[Y_t^2]}{dt} =  2(\eta^2 - \kappa)E[Y_t^2] + 2\kappa E[X_tY_t].
\end{equation}
Defining $E[Y_t^2]=z(t)$ and replacing the value for $E[X_tY_t]$ form equation \eqref{eq:solutionXY} we obtain an ODE
\begin{equation*}
z' = (\eta^2 - \kappa)z + 2\kappa be^{(2\mu+\sigma^2)\Delta t} + 2\kappa(X_{t-\Delta t}Y_{t-\Delta t} - b)e^{(\mu - \kappa + \sigma \eta \rho )\Delta t}.
\end{equation*}
Using again variation of parameters we obtain the following solution
\begin{eqnarray} 
z(t) = E[Y_t^2] = ce^{(2\mu+\sigma^2)\Delta t} + de^{(\mu -\kappa + \sigma\eta\rho )\Delta t} + (Y_{t-\Delta t}^2-c-d)e^{2(\eta^2 - \kappa)\Delta t}, 
\end{eqnarray}
with $c=\frac{2\kappa b }{2\mu+\sigma^2-2\eta^2 +2\kappa}$ and $d = \frac{2\kappa(X_{t-\Delta t}Y_{t-\Delta t} -b)}{\mu - 2\eta^2+\kappa+\sigma\eta\rho}$.

Now we are ready to approximate $E[\ln(Y_t)]$. From \eqref{eq:expec_logY_approx} we have
\begin{eqnarray}\nonumber
&& E[\ln(Y_t)] \approx \ln[E[Y_t]] -\frac{E[Y_t^2]}{2E[Y_t]^2}+ \frac{1}{2} = \ln\left( ae^{\mu \Delta t} + (Y_{t-\Delta t} -a)e^{-\kappa \Delta t}\right) + \frac{1}{2}\\
&& - \frac{ ce^{(2\mu+\sigma^2)\Delta t} + de^{(\mu -\kappa +\sigma\eta\rho)\Delta t}}{2(ae^{\mu \Delta t} + (Y_{t-\Delta t} - a)e^{-\kappa \Delta t})^2} -  \frac{(Y_{t-\Delta t}^2-c-d)e^{2(\eta^2 - \kappa)\Delta t}}{2(ae^{\mu \Delta t} + (Y_{t-\Delta t} - a)e^{-\kappa \Delta t})^2}
\end{eqnarray}
and 
\begin{eqnarray}
&&  E[r(Y_t)] = E[\ln(Y_t)] - \ln(Y_{t-\Delta t}) \approx  \ln\left( ae^{\mu \Delta t} + (Y_{t-\Delta t} -a)e^{-\kappa \Delta t}\right) + \frac{1}{2} -\\
& & \frac{ ce^{(2\mu+\sigma^2)\Delta t} + de^{(\mu -\kappa +\sigma\eta\rho)\Delta t}}{2(ae^{\mu \Delta t} + (Y_{t-\Delta t} - a)e^{-\kappa \Delta t})^2} - \frac{(Y_{t-\Delta t}^2-c-d)e^{2(\eta^2 - \kappa)\Delta t}}{2(ae^{\mu \Delta t} + (Y_{t-\Delta t} - a)e^{-\kappa \Delta t})^2} - \ln(Y_{t-\Delta t}) \nonumber
\end{eqnarray}
From \eqref{eq:variance_logY_appox} we have
 \begin{eqnarray} \nonumber
&& Var[r(Y_t)] = Var[\ln(Y_t)] \approx \frac{E[Y_t^2]}{E[Y_t]^2} - 1 = \\
&&\frac{ ce^{(2\mu+\sigma^2)\Delta t} + de^{(\mu+\sigma\eta\rho -\kappa)\Delta t}}{(ae^{\mu \Delta t} + (Y_{t-\Delta t} - a)e^{-\kappa \Delta t})^2} + \frac{(Y_{t-\Delta t}^2-c-d)e^{2(\eta^2 - \kappa)\Delta t}}{(ae^{\mu \Delta t} + (Y_{t-\Delta t} - a)e^{-\kappa \Delta t})^2} -1
 \end{eqnarray}
and  
 \begin{equation}
    Var[r(X_t)] = Var[\ln(X_t)] = \sigma^2 \Delta t.
 \end{equation}
From \eqref{eq:cov_logXY_approx} we obtain the covariance:
\begin{eqnarray}\nonumber
&Cov[r(X_t)r(Y_t)] = Cov[\ln(X_t)\ln(Y_t)] \approx \ln\left(\frac{E[X_tY_t]}{E[X_t]E[Y_t]} \right)  \\
&\approx \ln\left(\frac{be^{(2\mu+\sigma^2)\Delta t} + (X_{t-\Delta t}Y_{t-\Delta t} - b)e^{(\mu -\kappa + \sigma\eta\rho)\Delta t}}{aX_0e^{2\mu \Delta t} + (Y_{t-\Delta t}X_{t-\Delta t}- aX_{t-\Delta t})e^{(\mu-\kappa)\Delta t}}\right).
\end{eqnarray}

\section{Dimension reduction of 3-dim HJB \eqref{eq:HJB_original}} \label{appendix:HJB}
\noindent For ease of notation let $\tilde{\sigma} = \sigma^2 - 2\sigma\eta\rho +\eta^2$ and rewrite \eqref{eq:HJB_original}:
\begin{eqnarray}\nonumber
&&G_t + \sup_{\pi_1}\{\frac{1}{2}(\pi_1^2 \tilde{\sigma} v^2G_{vv} + \tilde{\sigma}z^2G_{zz} + 2\pi_1\tilde{\sigma}vzG_{vz})+ \\ 
&&(\pi_1[\mu - \kappa(z-1)])vG_v + (\mu + \eta^2  - \sigma\eta\rho - \kappa(z-1))zG_z\} =0. \label{eq:HJBrewrite}
\end{eqnarray}
The first order condition for the maximization is
\begin{equation}
\pi_1^*\tilde{\sigma}vG_{vv} + \tilde{\sigma}zG_{vz} + [\mu - \kappa(z-1)]G_v = 0.
\end{equation}
Now assuming $G_{vv}<0$ the first order condition is sufficient, yielding
\begin{equation}
  \pi_1^* = - \frac{\tilde{\sigma}zG_{vz} + [\mu - \kappa(z-1)]G_v}{\tilde{\sigma}vG_{vv}}. \label{eq:optimWeightFirstOrder}
\end{equation}
Replacing \eqref{eq:optimWeightFirstOrder} back into \eqref{eq:HJBrewrite} yields:
\begin{eqnarray*}\nonumber
&&G_t + \frac{1}{2} \{ \frac{(\tilde{\sigma}zG_{vz} + [\mu - \kappa(z-1)]G_v)^2}{\tilde{\sigma}^2v^2G^2_{vv}}\tilde{\sigma}v^2G_{vv} + \tilde{\sigma}z^2G_{zz} \\ \nonumber
&&- 2\frac{\tilde{\sigma}zG_{vz} + [\mu - \kappa(z-1)]G_v}{\tilde{\sigma}vG_{vv}}\tilde{\sigma}vzG_{vz} \} + \left[\mu +\eta^2 -\sigma\eta\rho -\kappa(z-1)\right]zG_z \\\nonumber
&&- \frac{\tilde{\sigma}zG_{vz} + [\mu - \kappa(z-1)]G_v}{\tilde{\sigma}vG_{vv}}\left[\mu - \kappa(z-1)\right]vG_v   = 0.
\end{eqnarray*}
Multiplying both sides of equation by $\tilde{\sigma}G_{vv}$ we get:
\begin{eqnarray}\nonumber
&&\tilde{\sigma}G_tG_{vv} + \frac{1}{2}(\tilde{\sigma}zG_{vz} + [\mu -\kappa(z-1)]G_v)^2 - (\tilde{\sigma}zG_{vz} + [\mu - \kappa(z-1)]G_v)\tilde{\sigma}zG_{vz}\\ \nonumber
&&+ \frac{1}{2}\tilde{\sigma}z^2G_{zz}G_{vv} - (\tilde{\sigma}zG_{vz} + [\mu - \kappa(z-1)]G_v)[\mu -\kappa(z-1)]G_v \\ \nonumber
&&+ \tilde{\sigma}[\mu +\eta^2 -\sigma\eta\rho -\kappa(z-1)]zG_zG_{vz} = 0. \label{eq:HJBsimplified}
\end{eqnarray}
Expanding gives
\begin{eqnarray}\nonumber
&&\tilde{\sigma}G_tG_{vv} + \frac{1}{2}\tilde{\sigma}^2z^2G^2_{vz} + \frac{1}{2}[\mu -\kappa(z-1)]^2G^2_v + \tilde{\sigma}z[\mu -\kappa(z-1)]G_vG_{vz} \\ \nonumber
&&  -\tilde{\sigma}^2z^2G^2_{vz} - \tilde{\sigma}z[\mu -\kappa(z-1)]G_vG_{vz} +\frac{1}{2}\tilde{\sigma}z^2G_{zz}G_{vv} - \tilde{\sigma}z[\mu -\kappa(z-1)]G_vG_{vz}  \\ 
&& - [\mu -\kappa(z-1)]^2G^2_v + \tilde{\sigma}[\mu +\eta^2 -\sigma\eta\rho -\kappa(z-1)]zG_zG_{vz} = 0.
\end{eqnarray}
Which further simplifies to
\begin{eqnarray} \nonumber
&& \tilde{\sigma}G_tG_{vv} -\frac{1}{2}[\mu -\kappa(z-1)]^2G^2_v + \frac{1}{2}\tilde{\sigma}z^2G_{zz}G_{vv} - \tilde{\sigma}z[\mu -\kappa(z-1)]G_vG_{vz}  \\  
&& -  \frac{1}{2}\tilde{\sigma}^2z^2G^2_{vz} + \tilde{\sigma}[\mu +\eta^2 -\sigma\eta\rho -\kappa(z-1)]zG_zG_{vz} = 0. \label{eq:simplPDE}
\end{eqnarray}
At this stage we were able to turn our four variable PDE into three, but we can get eliminate one more. For this we consider the following separation ansatz:
\begin{equation}
G(t, v, z) = f(t,z)v^\gamma, \label{eq:firstAnsatz}
\end{equation}
with the terminal condition
\begin{equation}
f(T,z) = 1  \ \ \ \ \forall z.   
\end{equation}
We compute the derivatives of \eqref{eq:firstAnsatz}:
\begin{eqnarray*}
& G_t = f_tv^\gamma, \ \ \ G_v = fv^{\gamma -1}\gamma , \ \ G_z=f_zv^\gamma , \ \ \
G_{vv} = fv^{\gamma -2}\gamma (\gamma -1), \\
& G_{vz} = f_zv^{\gamma -1}\gamma , \ \ \ G_{zz} =v^\gamma f_{zz}.
\end{eqnarray*}
Replace derivative back into \eqref{eq:simplPDE} and divide by $v^{2(\gamma -1)}\gamma$ to get
\begin{eqnarray}\nonumber
&&\tilde{\sigma}(\gamma-1) ff_t  - \frac{1}{2}\tilde{\sigma}^2\gamma z^2f_z^2 - \frac{1}{2}\gamma[\mu -\kappa(z-1)]^2f +\frac{1}{2}\tilde{\sigma}(\gamma-1)z^2ff_{zz} -\\ 
&& \tilde{\sigma}\gamma [\mu -\kappa(z-1)]zff_z +  \tilde{\sigma}(\gamma -1)[\mu + \eta^2 -\sigma\eta\rho -\kappa(z-1)]ff_z=0. \label{eq:PDE_for_DGM}
\end{eqnarray} 
We now have a PDE with only two variables instead of four.




\pagebreak
\clearpage
\begin{table}[H]
\begin{center} 
\resizebox{0.7\linewidth}{!}{%
  \begin{tabular}{ | l  c |  }
    \hline
    Criterion &  Average portfolio return $\hat{R}(r_p)$\\ \hline
    MVC & 35\% \\ 
    SC & 61\% \\ 
    DS & 83\% \\ \hline
  \end{tabular}}
  \end{center}
  \caption{Average over 500 simulations of portfolio returns at terminal time $T$ (day 1000) with dynamic switching (DS) is higher than average portfolio return with only stochastic control (SC) or only mean-variance-criterion (MVC).} \label{tab:Table_Averag}  
\end{table}
\pagebreak
\clearpage
\begin{itemize}
    \item[Figure \ref{fig:CointelationWithRhoM1}:] (Up) Simulated path of cointelation model \eqref{eq:modelCointel} with $\rho=-1$, $\theta=0.1$, $\sigma=0.01$; (Down) Corresponding measured correlation \eqref{eq:measuredCorr} as a function of the time increment increases from $-1$ to $1$.
    \item[Figure \ref{fig:birdEyeDGM}:] Bird’s-eye perspective of overall DGM architecture \cite{SolvingPDEwithDL}.
    \item[Figure \ref{fig:detailedDGM}:] Operations within a single DGM layer \cite{SolvingPDEwithDL}
    \item[Figure \ref{fig:mertonNNAnalitycal}:] Analytical solution of the Merton Problem.
    \item[Figure \ref{fig:mertonNN}] Approximate solution of Merton problem using DGM.
    \item[Figure \ref{fig:mertonerror}:] Error between analytical and approximate solution of Merton problem.
    \item[Figure \ref{fig:changeInRhoMu}:] Approximate solutions to PDE \eqref{eq:PDE_for_DGM} with DGM for four different scenarios of $\rho$ and $\mu$ and fixed $\sigma=0.2$, $\eta=0.19$, $\gamma=0.5$.
    \begin{itemize}
        \item[(a)] Approximate solution with low $\mu = 0.01$ and low $\rho=-0.5$.
        \item[(b)] Approximate solution with low $\mu = 0.01$ and high $\rho = 0.5$.
        \item[(c)] Approximate solution with high $\mu = 0.4$ and low $\rho = -0.5$.
        \item[(d)] Approximate solution with high $\mu = 0.4$ and high $\rho =0.5$.
    \end{itemize}
    \item[Figure \ref{fig:GaussianMixture}:] Two examples of Gaussian Mixture Simulations with different number of bands. 
    \begin{itemize}
        \item[(a)] Empirical distribution of random variable sampled from cointelation model \eqref{eq:modelCointel} in three different zones described in Figure \ref{fig:BplusBminus}.
        \item[(b)] Empirical distribution of random variable sampled from cointelation model \eqref{eq:modelCointel} in five different zones: two additional zones were added to the initial three zones in Figure \ref{fig:BplusBminus}.
    \end{itemize}
    \item[Figure \ref{fig:MLvsDS_1path}:] (a) one simulated scenario based on cointelation model \eqref{eq:modelCointel} with parameters: $\mu=0.05, \sigma=0.17, \eta=0.16, \kappa=0.1, \rho = -0.6$ and scaled spread: $\kappa(X_t - Y_t) $; (b)  portfolio return and optimal weight of asset $X$ with Dynamic Switching approach; (c) portfolio return and optimal weight of asset $X$ and $Y$ with Machine Learning approach.
    \item[Figure \ref{fig:SC_vs_ML}:] Histogram of excess ($P\&L$) for $ML_{LS}$ vs SC at terminal time $T$.
    \item[Figure \ref{fig:DS_vs_ML}:] Histogram of excess ($P\&L$) for ML vs FM at terminal time $T$.
    \item[Figure \ref{fig:DS_spread_weights}:] In FM approach optimal long/short strategies are more volatile than optimal long only strategies.
    \item[Figure \ref{fig:ML_spread_weights}:] In ML approach optimal long/short strategies are slightly more volatile than optimal long only strategies.
\end{itemize}

\pagebreak
\clearpage
\begin{figure*}[htb!]
\begin{center}
\includegraphics[width=1.0\textwidth]{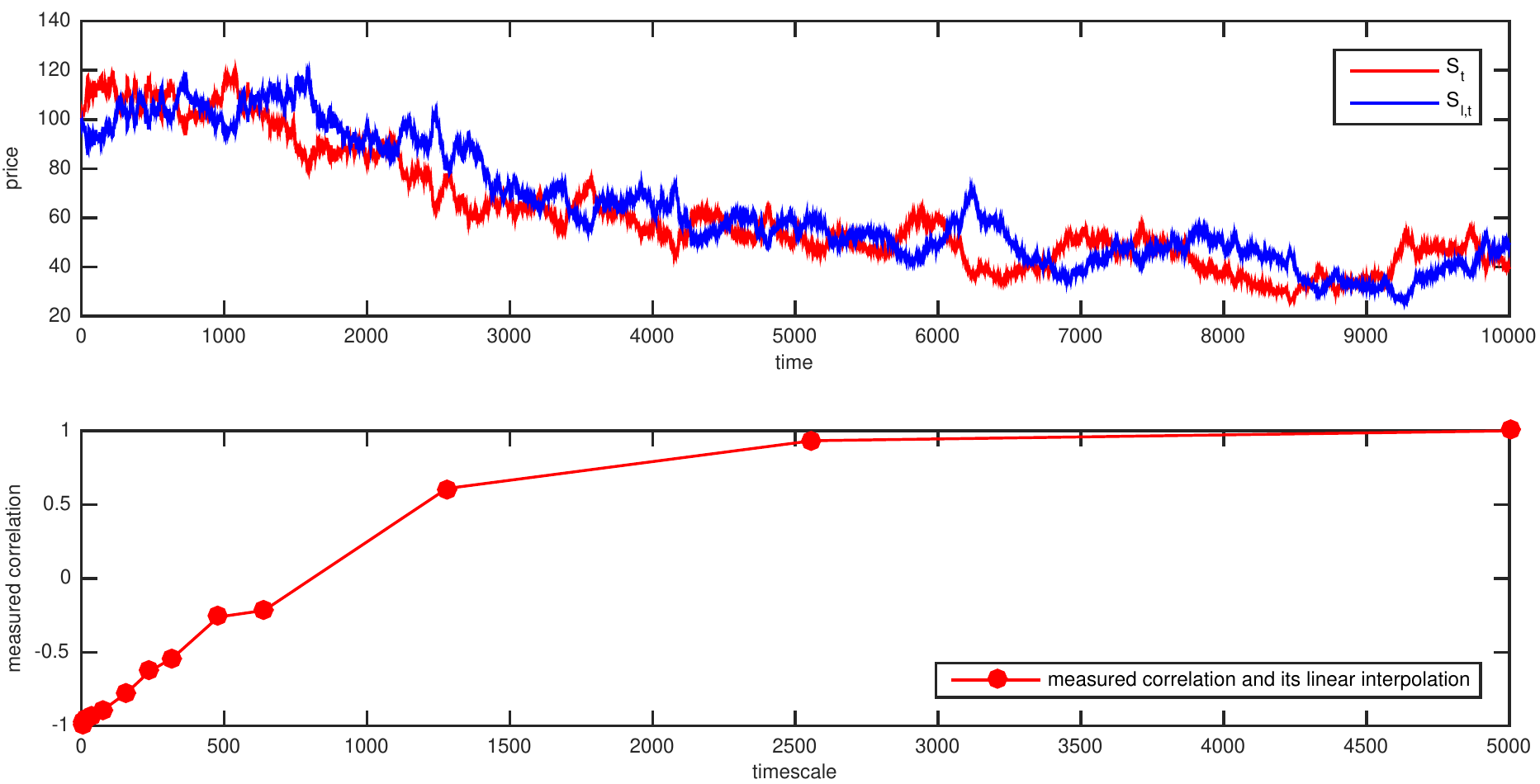}
\caption{}
\label{fig:CointelationWithRhoM1}
\end{center}
\end{figure*}

\pagebreak
\clearpage
\begin{figure*}[htb!]
\begin{center}
\includegraphics[width=0.7\textwidth]{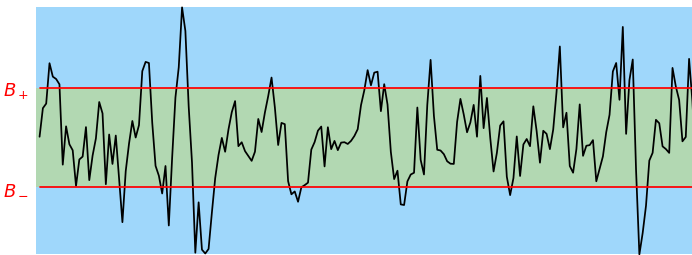}
\caption{}
\label{fig:BplusBminus}
\end{center}
\end{figure*}

\pagebreak
\clearpage
\begin{figure*}[htb!]
    \centering
    \includegraphics{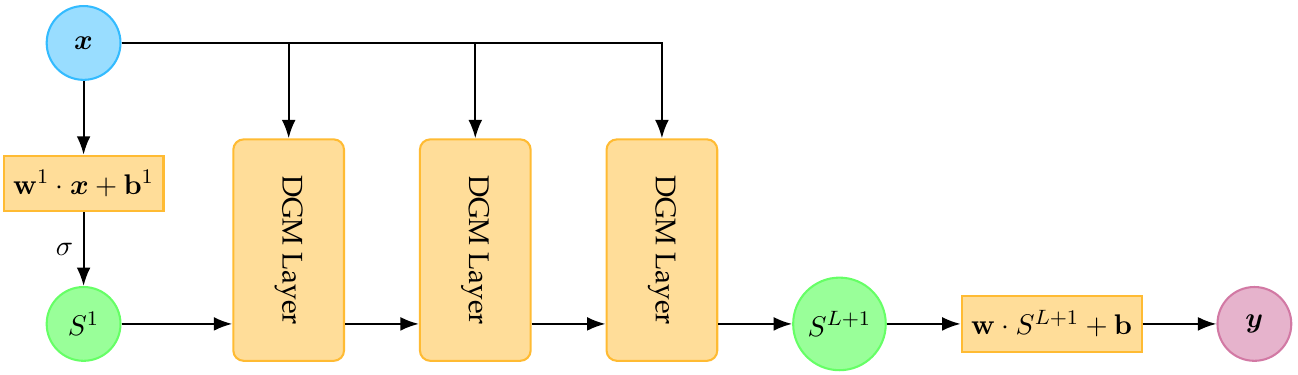}
    \caption{} 
    \label{fig:birdEyeDGM}
\end{figure*}

\pagebreak
\clearpage
\begin{figure*}[htb!]
    \centering
    \includegraphics{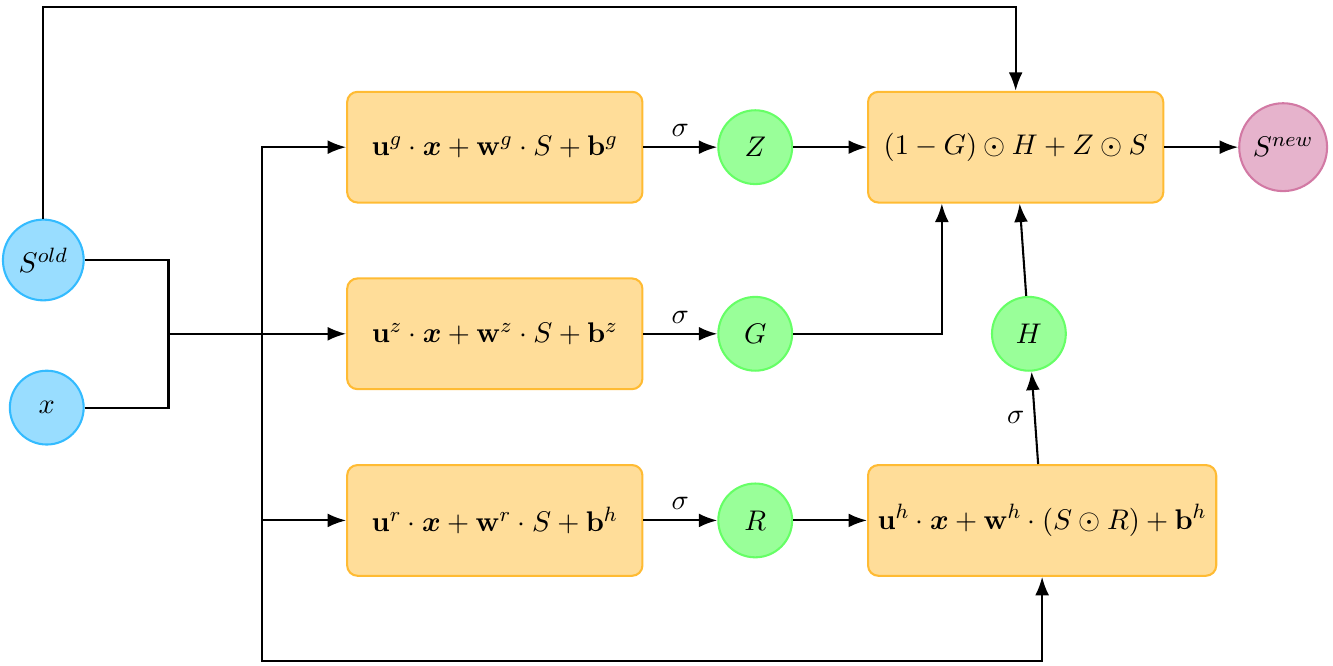}
    \caption{}
    \label{fig:detailedDGM}
\end{figure*}

\pagebreak
\clearpage
\begin{figure}[htb]
\begin{center}
\includegraphics[width=0.6\textwidth]{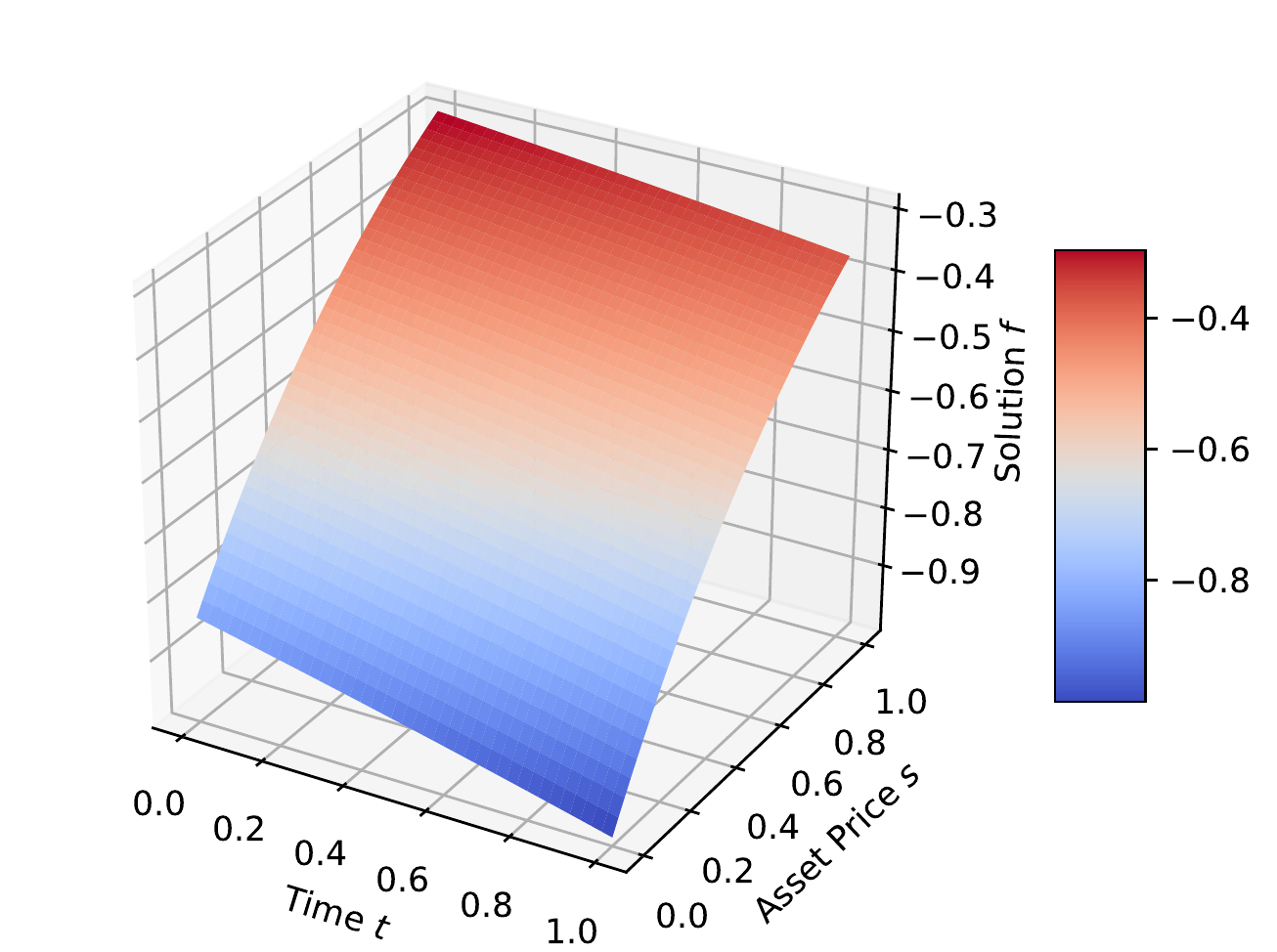}
\caption{}
\label{fig:mertonNNAnalitycal}
\end{center}
\end{figure}

\pagebreak
\clearpage
\begin{figure}[htb]
\begin{center}
\includegraphics[width=0.6\textwidth]{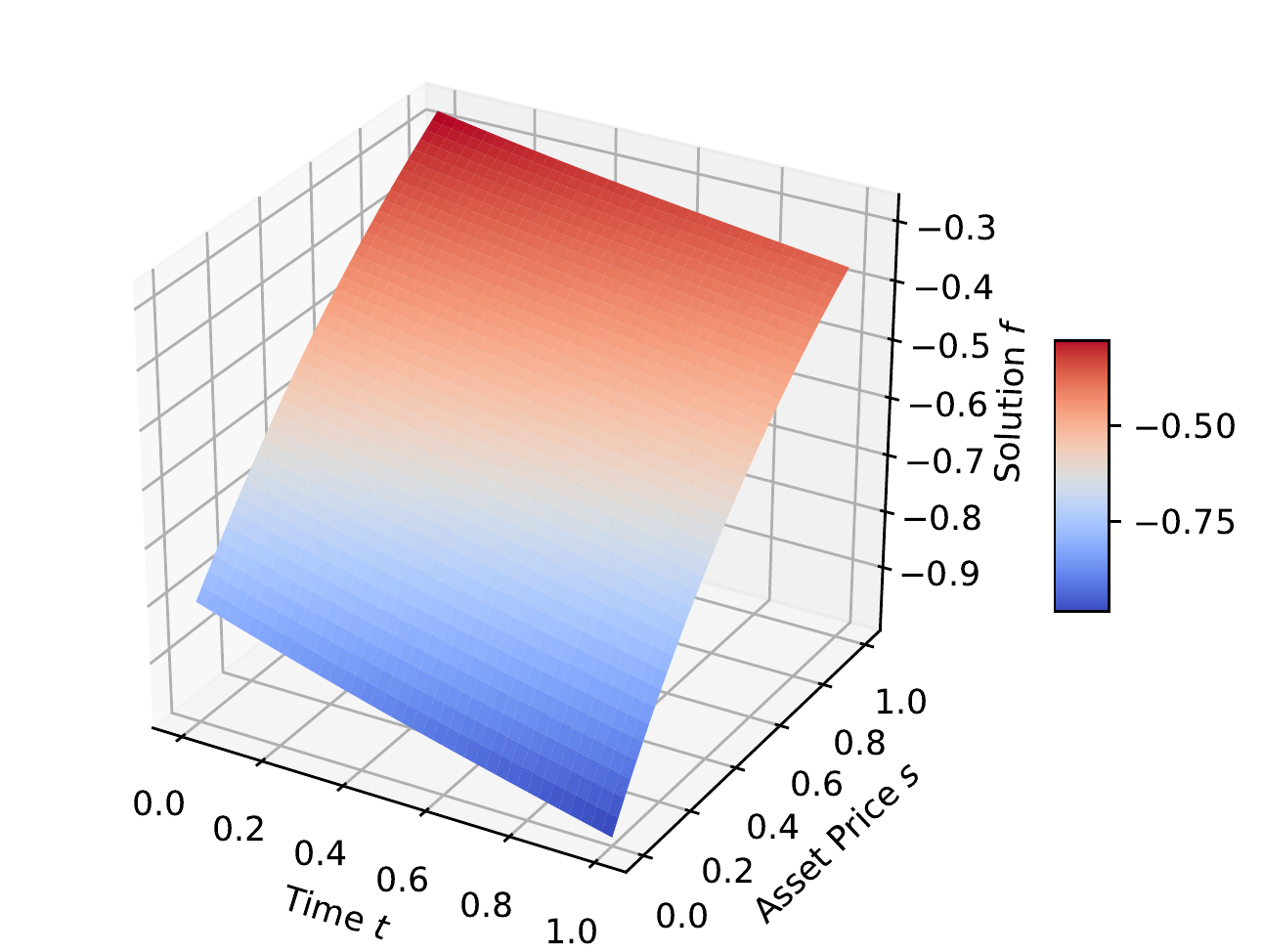}
\caption{}
\label{fig:mertonNN}
\end{center}
\end{figure}

\pagebreak
\clearpage
\begin{figure}[htb]
\begin{center}
\includegraphics[width=0.6\textwidth]{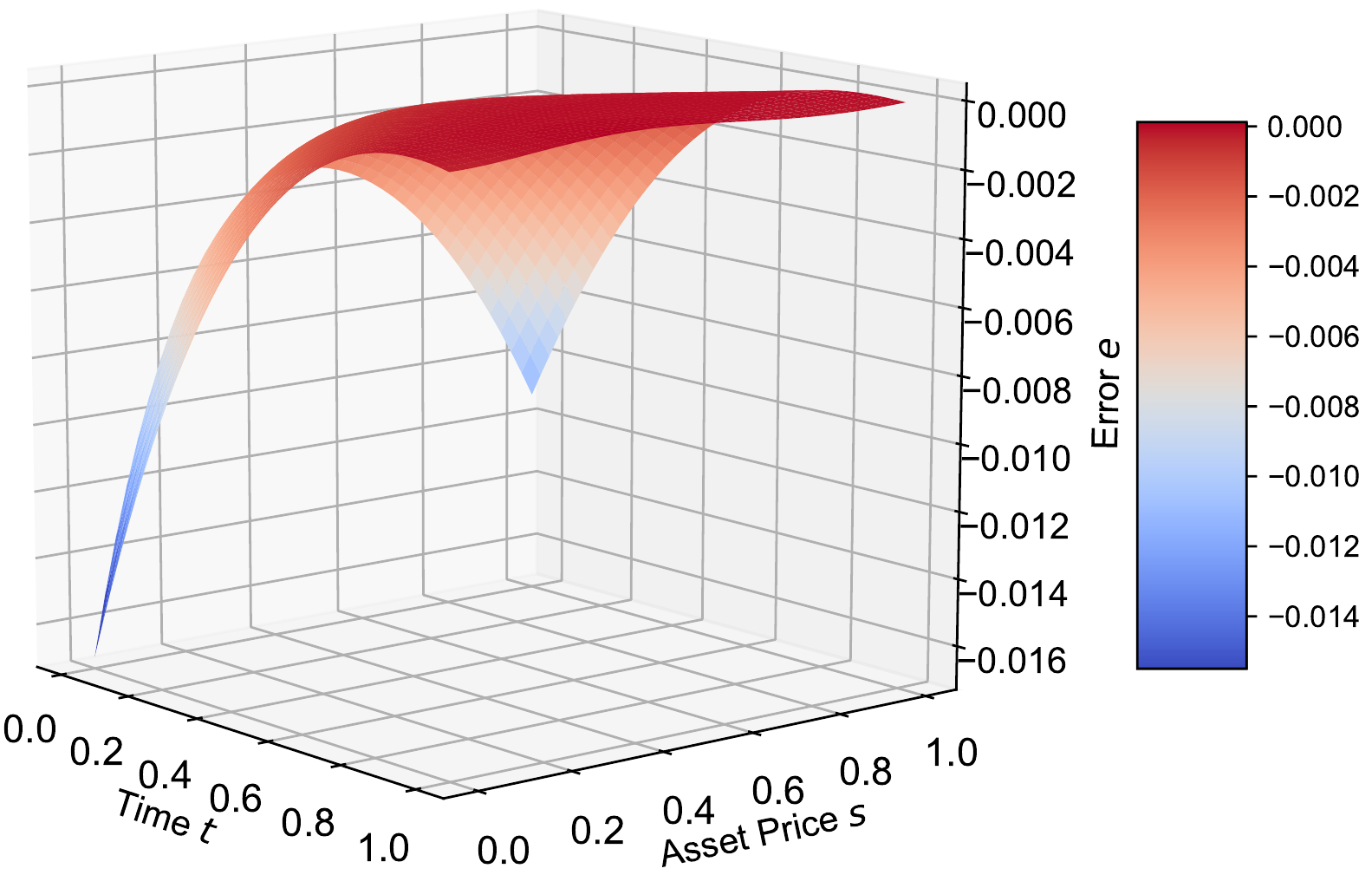}
\caption{}
\label{fig:mertonerror}
\end{center}
\end{figure}

\pagebreak
\clearpage
\begin{figure*}
\centering
\subcaptionbox{}{\includegraphics[width=0.45\linewidth]{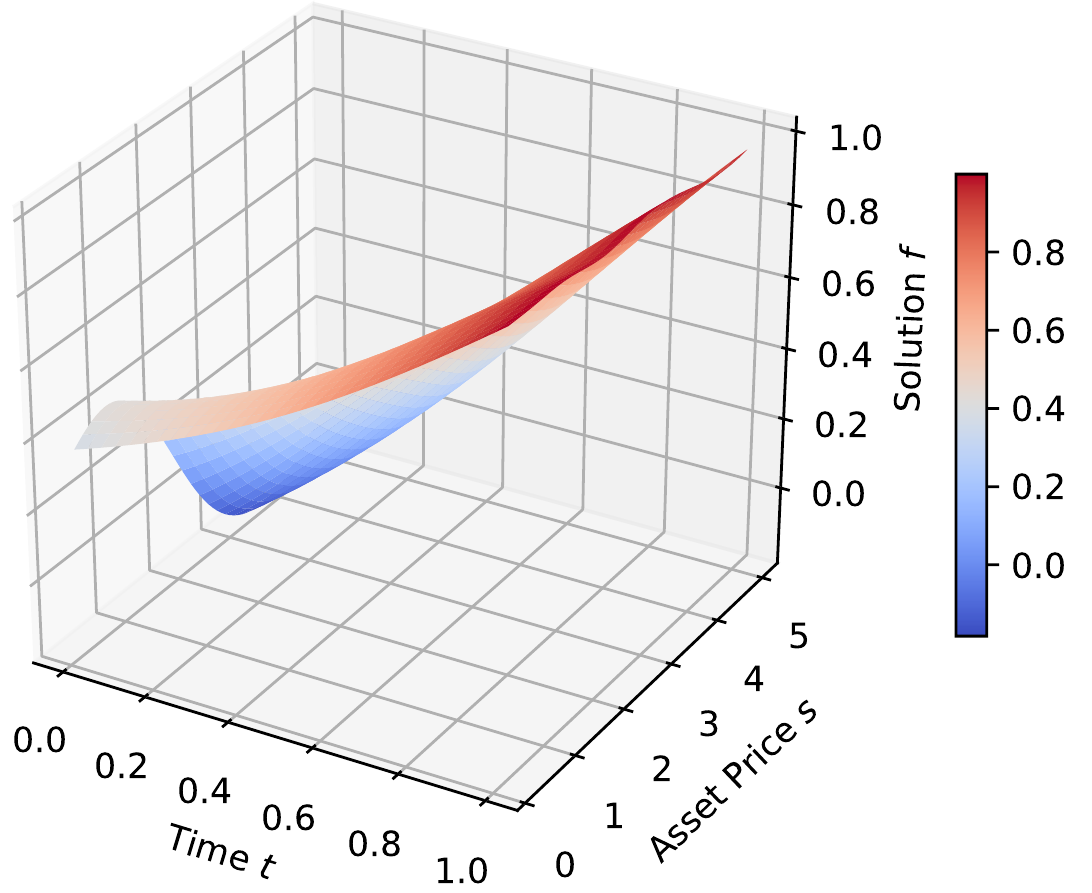}}
\subcaptionbox{}{\includegraphics[width=0.45\linewidth]{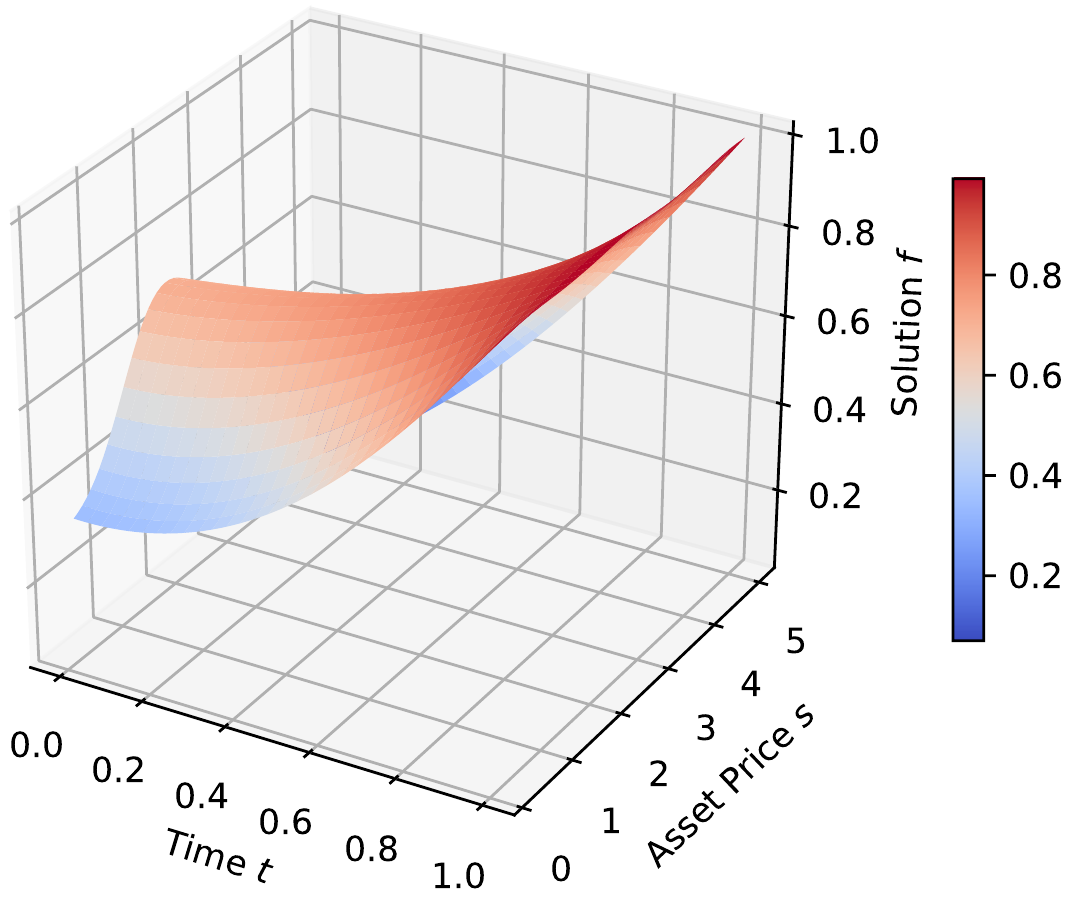}}
\subcaptionbox{}{\includegraphics[width=0.45\linewidth]{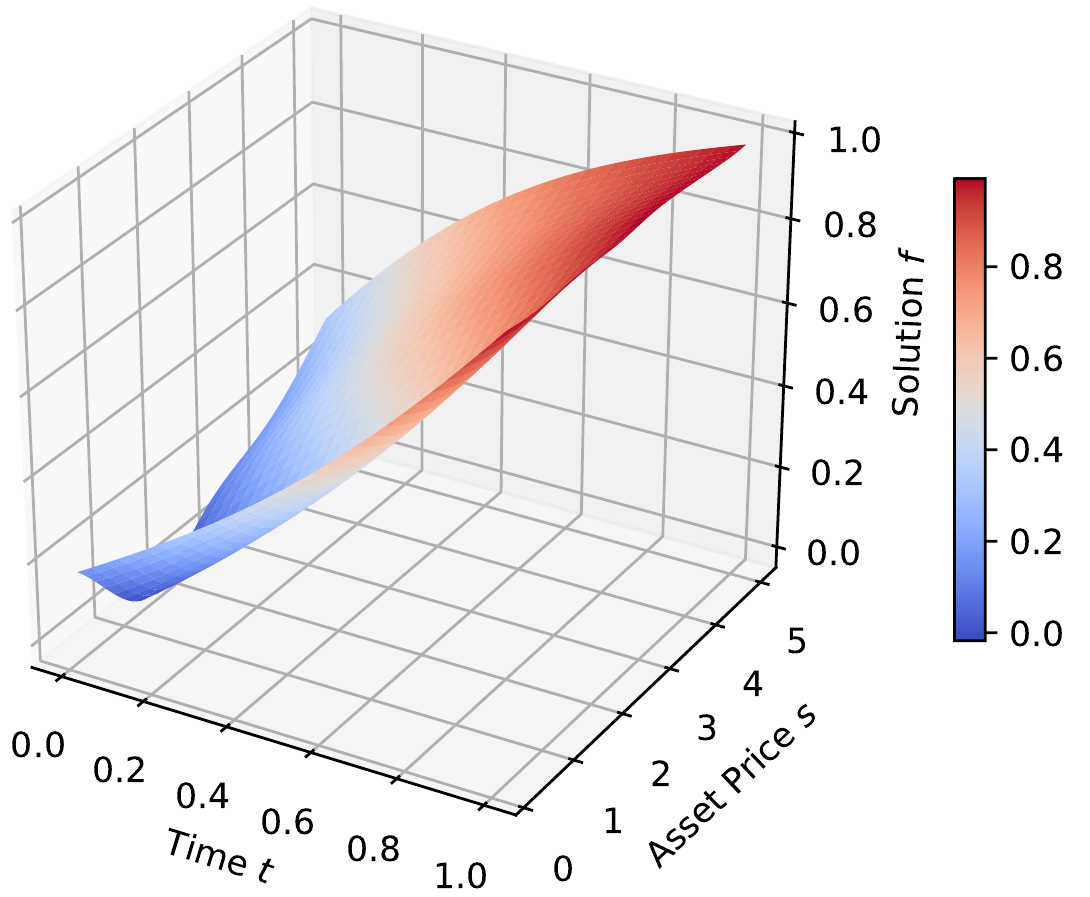}}
\subcaptionbox{}{\includegraphics[width=0.45\linewidth]{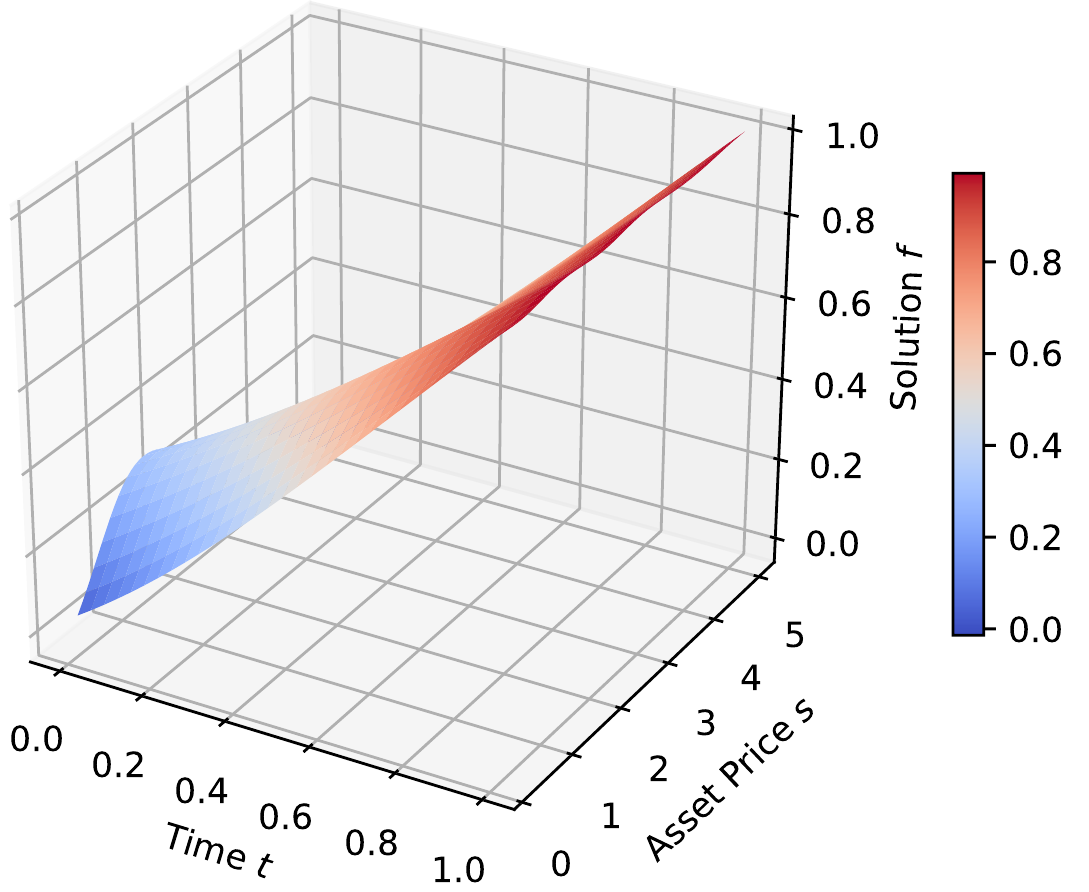}}
\caption{} \label{fig:changeInRhoMu}
\end{figure*}

\pagebreak
\clearpage
\begin{figure*}
\centering
\subcaptionbox{}
{\includegraphics[height=15cm, width=0.49\linewidth]{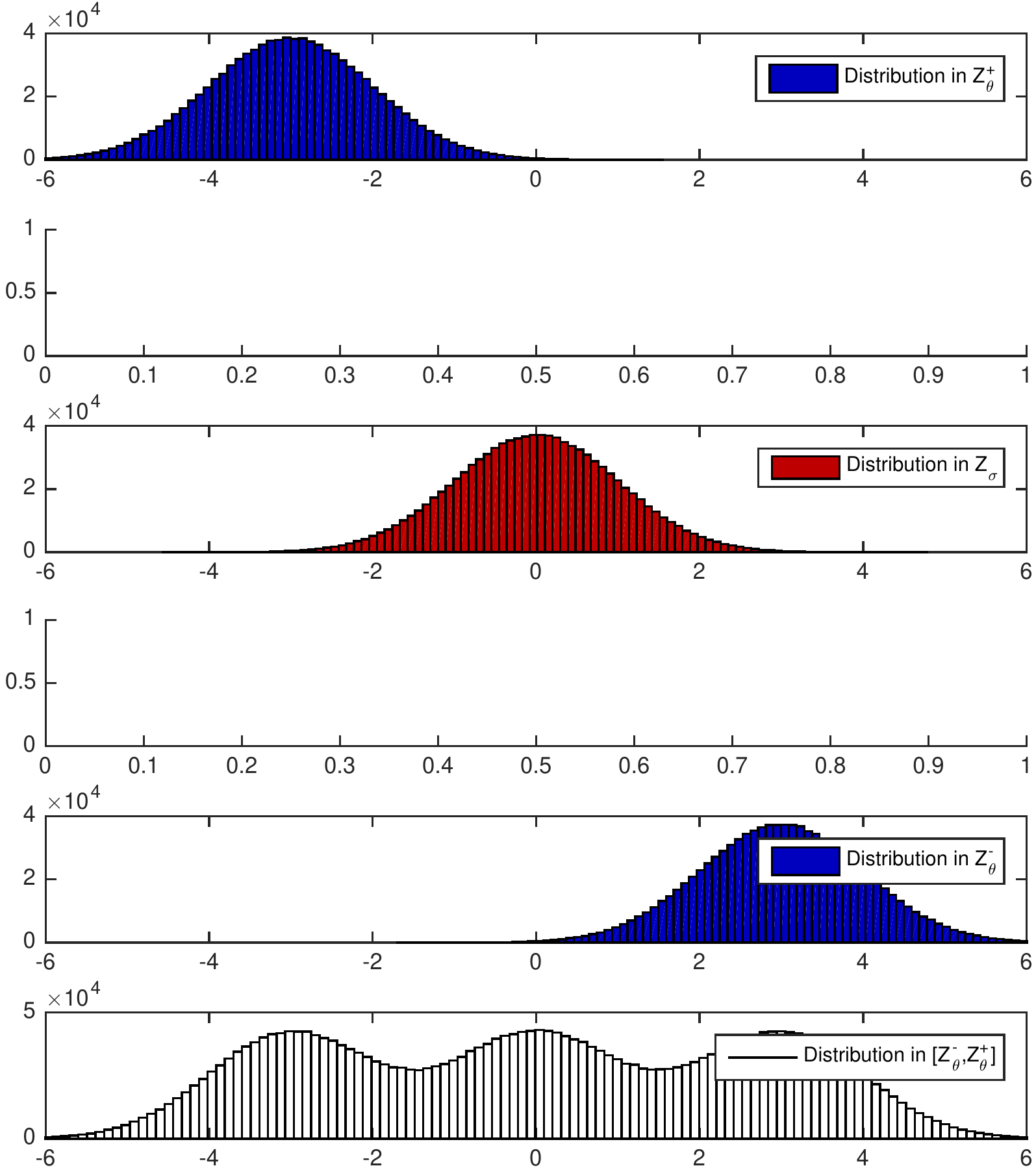}}
\subcaptionbox{}
{\includegraphics[height=15cm, width=0.49\linewidth]{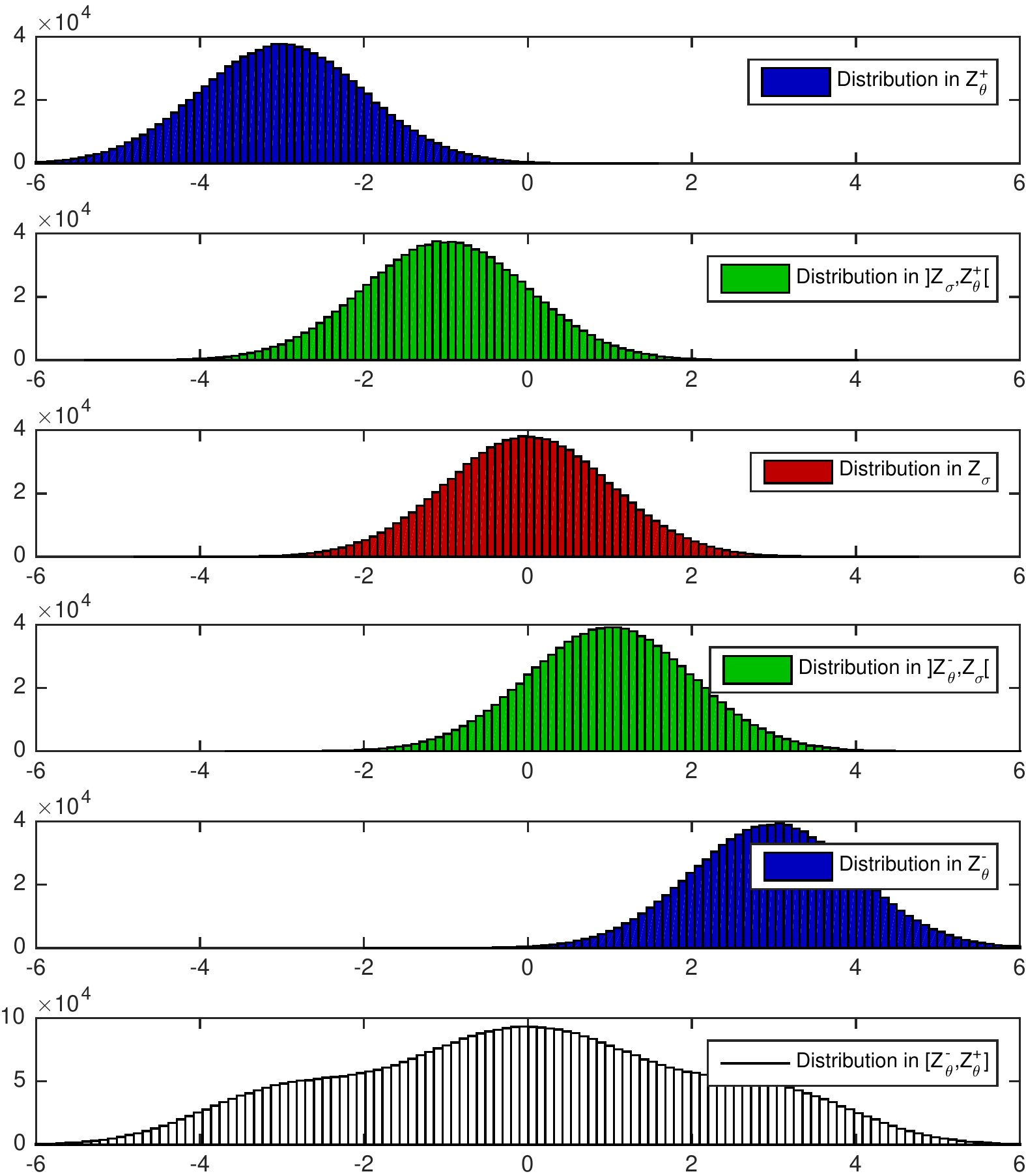}}
\caption{} \label{fig:GaussianMixture}
\end{figure*}

\pagebreak
\clearpage
\begin{figure*}[htb!]
\begin{center}
\includegraphics[width=1.0\textwidth]{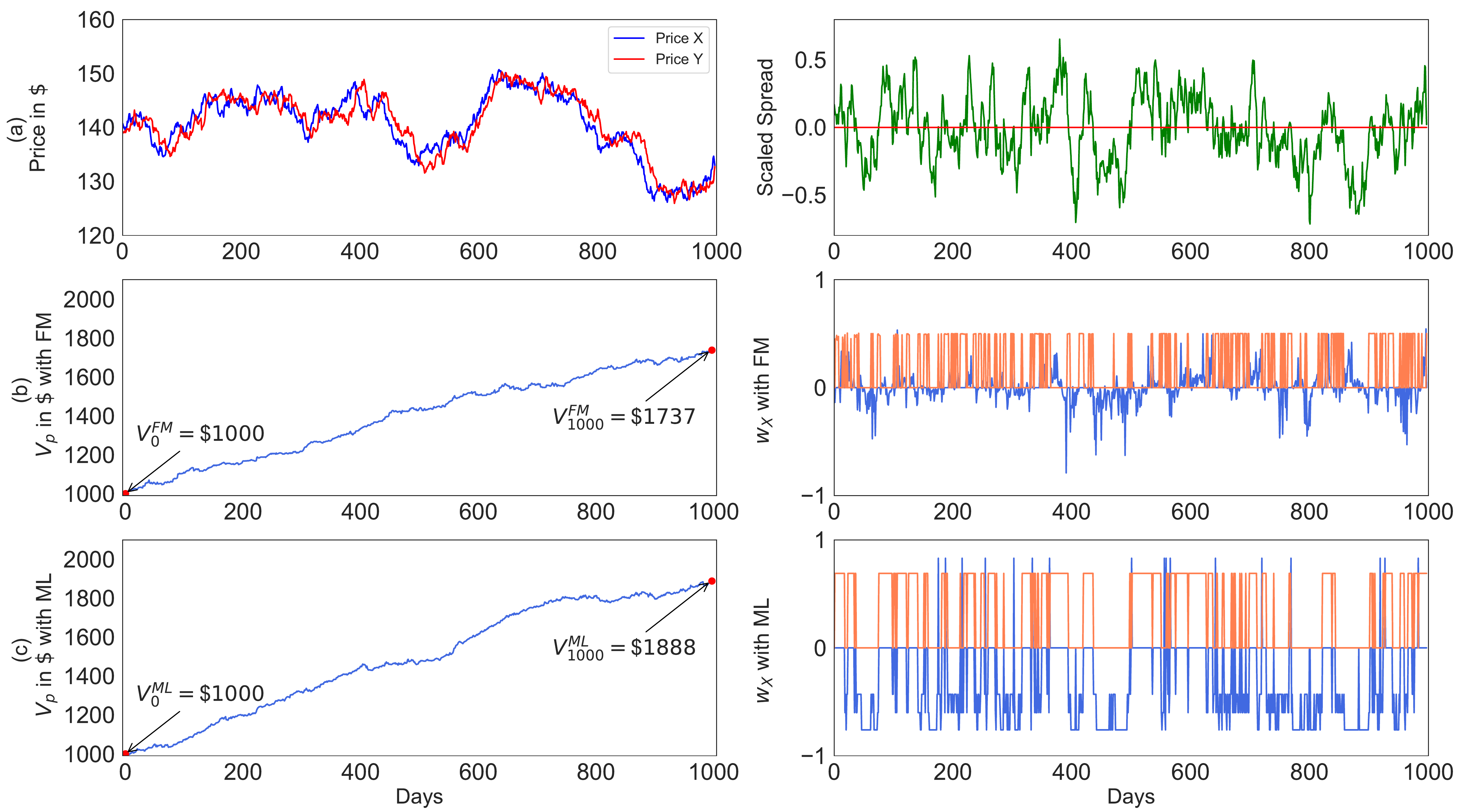}
\caption{}
\label{fig:MLvsDS_1path}
\end{center}
\end{figure*}

\pagebreak
\clearpage
\begin{figure*}[htb!]
\begin{center}
\includegraphics[width=1.0\linewidth]{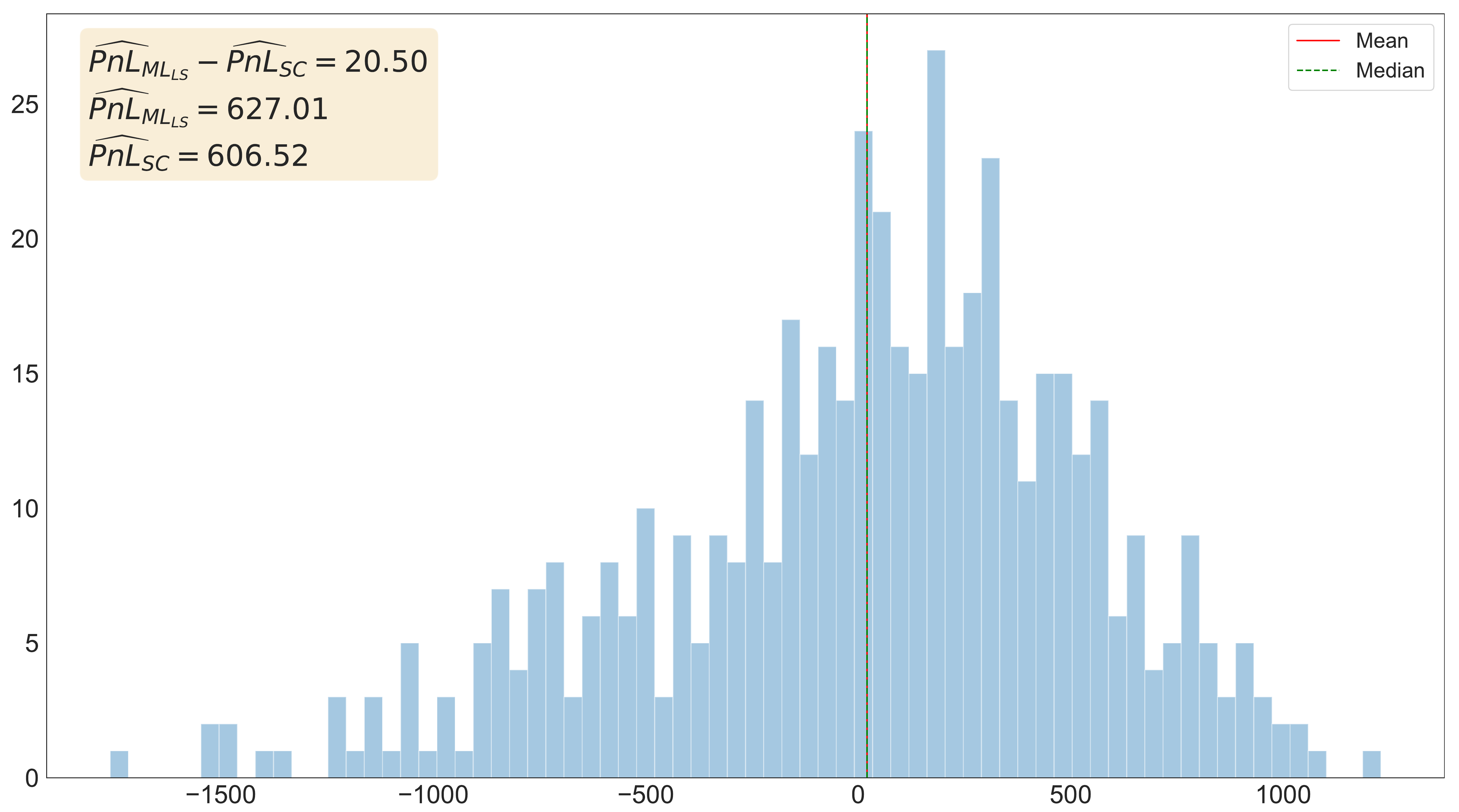}
\caption{}
\label{fig:SC_vs_ML}
\end{center}
\end{figure*}

\pagebreak
\clearpage
\begin{figure*}[htb!]
\begin{center}
\includegraphics[width=1.0\linewidth]{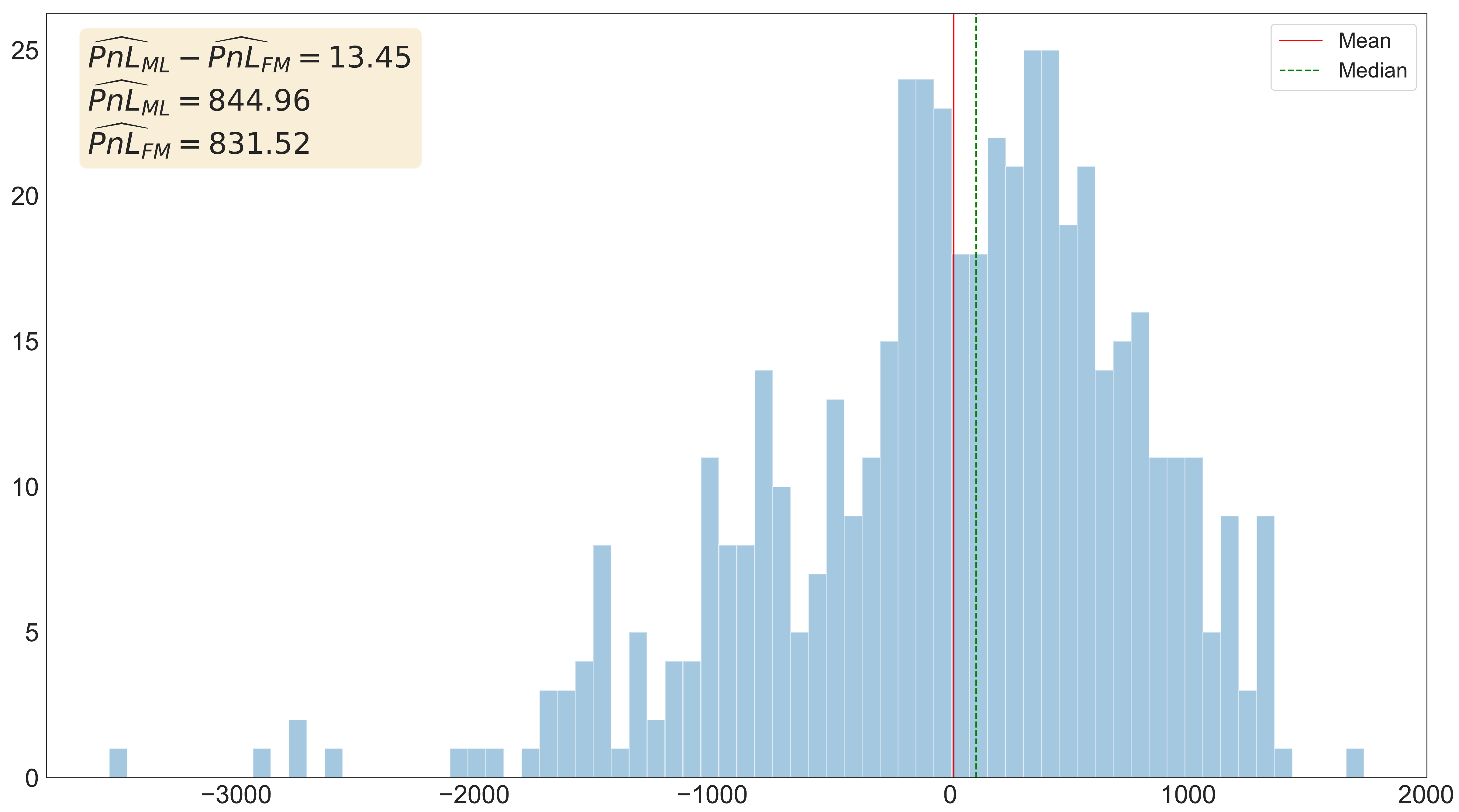}
\caption{} 
\label{fig:DS_vs_ML}
\end{center}
\end{figure*}

\pagebreak
\clearpage
\begin{figure*}[htb!]
\begin{center}
\includegraphics[width=1.0\linewidth]{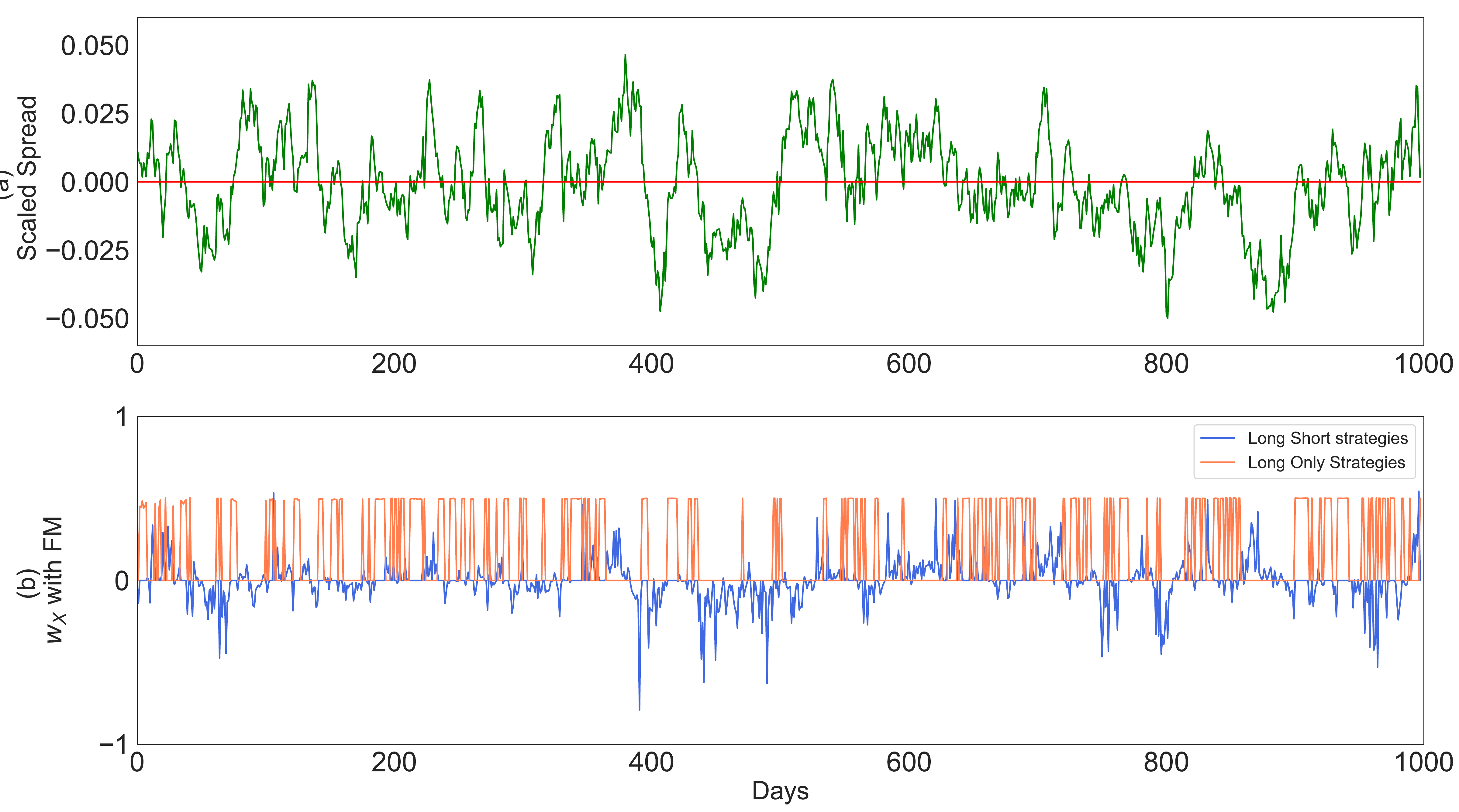}
\caption{}
\label{fig:DS_spread_weights}
\end{center}
\end{figure*}

\pagebreak
\clearpage
\begin{figure*}[htb!]
\begin{center}
\includegraphics[width=1.0\linewidth]{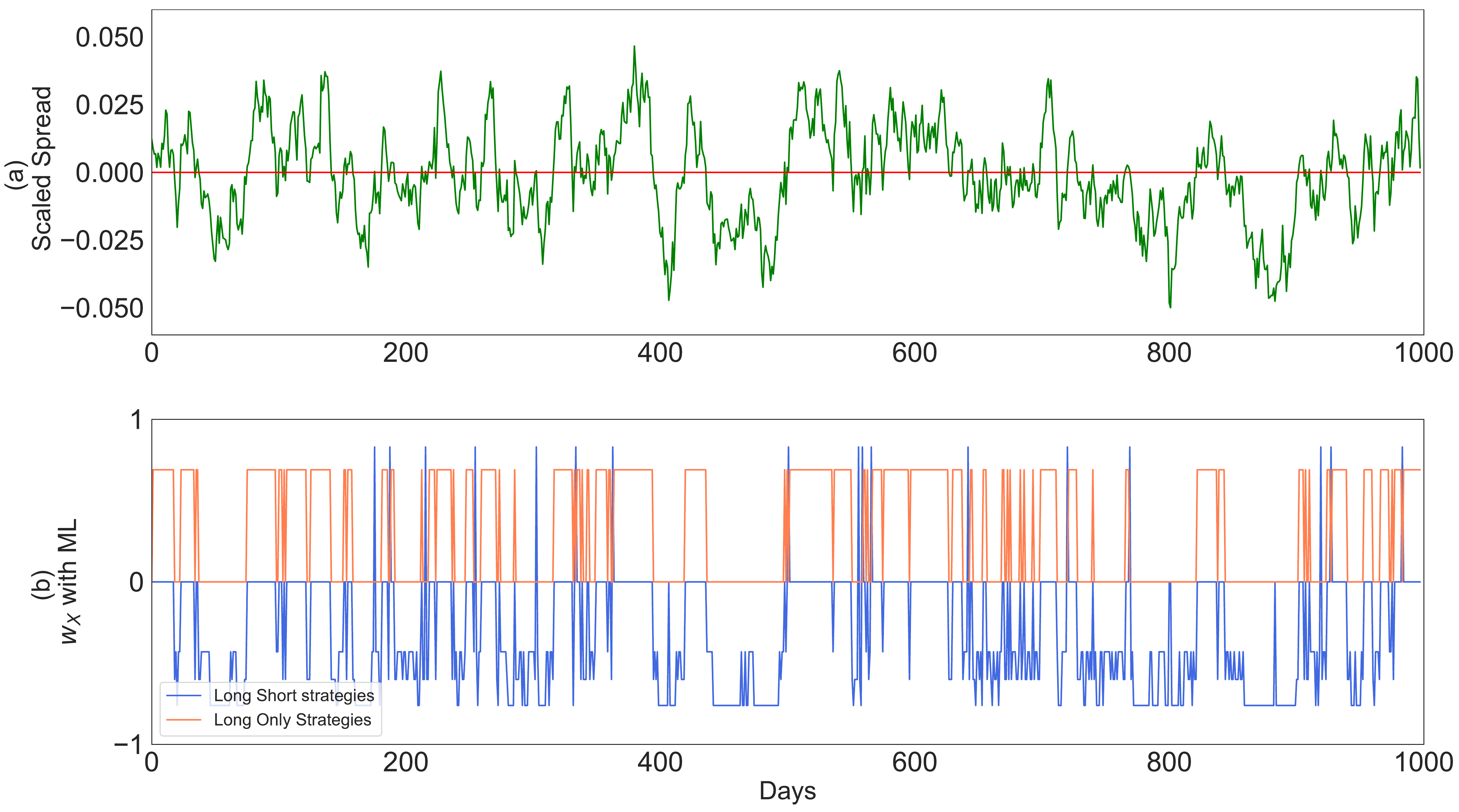}
\caption{}
\label{fig:ML_spread_weights}
\end{center}
\end{figure*}

\pagebreak
\clearpage

\end{document}